\documentclass[sigconf]{acmart}

\author{Jianfeng Deng}
\affiliation{%
  \institution{Guangxi University}
  \city{Nanning}
  \country{China}
}
\email{jianfeng_web@163.com}

\author{Qingfeng Chen}
\authornote{Corresponding author}
\affiliation{%
  \institution{Guangxi University}
  \city{Nanning}
  \country{China}
}
\email{qingfeng@gxu.edu.cn}

\author{Debo Cheng}
\authornotemark[1]
\affiliation{%
  \institution{University of South Australia}
  \city{Adelaide}
  \country{Australia}
}
\email{debo.cheng@unisa.edu.au}

\author{Jiuyong Li}
\affiliation{%
  \institution{University of South Australia}
  \city{Adelaide}
  \country{Australia}
}
\email{jiuyong.li@unisa.edu.au}

\author{Lin Liu}
\affiliation{%
  \institution{University of South Australia}
  \city{Adelaide}
  \country{Australia}
}
\email{liu.lin@unisa.edu.au}

\author{Xiaojing Du}
\affiliation{%
  \institution{University of South Australia}
  \city{Adelaide}
  \country{Australia}
}
\email{xiaojing.du@mymail.unisa.edu.au}

\setcopyright{none}

\renewcommand\footnotetextcopyrightpermission[1]{} 

\DeclareMathAlphabet{\mathcal}{OMS}{cmsy}{m}{n}  

\usepackage{hyperref}
\usepackage{url}
\usepackage{graphicx}
\usepackage{caption}

\usepackage{diagbox}
\usepackage{multirow}
\usepackage{booktabs}
\usepackage{array}
\usepackage{enumitem}

\usepackage{amsthm}

\usepackage{amssymb} 
\usepackage{amsmath}

\usepackage{amsfonts}

\newtheorem{assumption}{Assumption}		%
\newtheorem{theorem}{Theorem}		%
\newtheorem{lemma}{Lemma}		%
\newtheorem{definition}{Definition}

\usepackage[ruled,linesnumbered]{algorithm2e}

\usepackage{cleveref}

\crefname{equation}{equation}{equations}
\crefformat{equation}{(#2#1#3)}

\newcommand{\indep}{\perp\!\!\!\perp}
\newcommand{\nindep}{\not\!\perp\!\!\!\perp}

\begin{document}

\title{Mitigating Dual Latent Confounding Biases in Recommender Systems}

\begin{abstract}
Recommender systems are extensively utilised across various areas to predict user preferences for personalised experiences and enhanced user engagement and satisfaction. Traditional recommender systems, however, are complicated by confounding bias, particularly in the presence of latent confounders that affect both item exposure and user feedback. Existing debiasing methods often fail to capture the complex interactions caused by latent confounders in interaction data, especially when dual latent confounders affect both the user and item sides. To address this, we propose a novel debiasing method that jointly  integrates the \underline{I}nstrumental \underline{V}ariables (IV) approach and  \underline{i}dentifiable Variational Auto-Encoder (iVAE)  for \underline{D}ebiased representation learning in \underline{R}ecommendation systems, referred to as IViDR. Specifically, IViDR leverages the embeddings of user features as IVs to address confounding bias caused by latent confounders between items and user feedback, and reconstructs the embedding of items to obtain debiased interaction data. Moreover, IViDR employs an Identifiable Variational Auto-Encoder (iVAE) to infer identifiable representations of latent confounders between item exposure and user feedback from both the original and debiased interaction data. Additionally, we provide theoretical analyses of the soundness of using IV and the identifiability of the latent representations.
Extensive experiments on both synthetic and real-world datasets demonstrate that IViDR outperforms state-of-the-art models in reducing bias and providing reliable recommendations.
\end{abstract}

 
\maketitle

\section{Introduction}
\label{sec:intro}
Recommender systems are widely employed to offer personalised suggestions across various domains, including video streaming \cite{gao2022kuairand}, e-commerce \cite{36}, and online search \cite{32}. These systems have evolved significantly with the advances in collaborative filtering techniques (e.g., Matrix Factorisation (MF) \cite{12}), deep learning algorithms (e.g., Deep Factorisation Machine (DeepFM) \cite{deepfm}, Deep Cross Network (DCN) \cite{dcn}, Deep Interest Network (DIN) \cite{din}), and graph neural networks (GNNs) (e.g., LightGCN \cite{lightgcn}). Despite these innovations, many traditional recommendation algorithms rely heavily on statistical associations, which can introduce estimation biases such as popularity bias~\cite{popularity1}, selection bias~\cite{select1,15}, and conformity bias~\cite{conformity1}. To address these challenges, some causality-based recommender systems have been developed and successfully deployed in real-world applications~\cite{ROBINS19861393}.

Causal recommender systems address estimation bias by incorporating causal inference techniques. Previous studies have employed classical causal inference methods \cite{20,21}, such as back-door adjustment \cite{20} and inverse propensity reweighting (IPW) \cite{23} to mitigate specific biases. For instance, IPS \cite{23} leverages IPW to address selection bias \cite{15}, while D2Q \cite{34} applies back-door adjustments to counter video duration bias, which typically skews recommendations towards longer videos. Additionally, the Popularity De-biasing Algorithm (PDA) \cite{35} specifically uses causal interventions to eliminate popularity bias and improve recommendation accuracy. However, these approaches generally overlook latent confounders, which can still significantly impact recommendation accuracy and fairness~\cite{20}.

Recent developments in causal recommendation methods have focused on accounting for latent confounders~\cite{idcf,23}, aiming to enhance the robustness and unbiasedness of estimations by modeling and adjusting latent confounders that influence both user behaviour and item exposure. Latent confounders, which are unobservable by nature, introduce substantial estimation bias in recommender systems~\cite{20,cheng2022local,cheng2024data}. For example, high-quality products are usually priced higher and tend to receive more positive ratings from users, leading to a false correlation between high prices and positive ratings. As a result, recommender systems may tend to recommend high-priced products, but high-priced products may be disliked by users. Previous works, such as Deconfounder \cite{29}, approximate latent confounders using historical user data to mitigate their effects. The iDCF~\cite{idcf} method employs proximal causal inference techniques to address the impact of latent confounders. However, these methods may fail to infer latent confounders directly from interaction data, as some latent confounders may lack proxy information within the data. As a result, the inference of latent confounders from such data can potentially be inaccurate.

Instrumental variables (IVs) have been used to decompose input vectors within recommendation models to mitigate the effects of latent confounders~\cite{cheng2024data,iv-4,iv-11,iv-29}, as demonstrated in methods like IV4Rec~\cite{IV-serach-data} and IV4Rec+~\cite{si2023enhancing}. While these approaches can be effective for debiasing, they need to modify the input vectors using representation decomposition, which does not have an identification guarantee for the learned representations. Consequently, different learned representations can yield varying performance outcomes, making it difficult to ascertain whether improvements in accuracy are truly due to the mitigation of latent confounders or simply variations in the input data representations~\cite{9}. Furthermore, the decomposed representations used in IV4Rec~\cite{IV-serach-data} and IV4Rec+~\cite{si2023enhancing} primarily address latent confounders between items and user feedback, while overlooking the latent confounders between item exposure and user feedback.

Recommender systems deal with complex user-item historical data, where latent confounders can arise from different sources. Broadly, latent confounders fall into two categories: those that can be inferred through proxy variables and those without reliable proxy variables. For example, a user's socioeconomic status can be estimated by the average price of products they have purchased, whereas a user's mood when providing ratings lacks a clear proxy variable.
Existing methods typically address only one type of latent confounder.  
For example, iDCF~\cite{idcf} primarily focuses on latent confounders with proxy variables, while methods like IV4Rec~\cite{IV-serach-data} and IV4Rec+~\cite{si2023enhancing} handle latent confounders without proxy variables but may overlook those with proxy variables, such as latent confounders between item exposure and user feedback.
Since IV4Rec~\cite{IV-serach-data} does not have an identification guarantee, simply combining IV4Rec and iDCF to address both types of latent confounders still fails to achieve identifiability on learned representation.
To the best of our knowledge, no practical solution currently exists that addresses both types of latent confounders (i.e., dual latent confounding biases) in recommender systems.

To address this challenge, we propose a novel debiasing method (IViDR) that jointly integrates the \underline{I}nstrumental \underline{V}ariables (IV) approach and \underline{i}dentifiable VAE for \underline{D}ebiased representation learning in \underline{R}ecommendation systems to mitigate dual latent confounding biases. Specifically, IViDR first utilises user feature embeddings as IVs to reconstruct the treatment variables and generate debiased interaction data for addressing the effects of between items and user feedback. Subsequently, IViDR employs an iVAE~\cite{9} to infer identifiable latent representations from a combination of additional observable proxy variables (e.g., user’s consumption level), interaction data, and debiased interaction data. Finally, IViDR adjusts for the inferred latent representations to mitigate confounding bias between item exposure and user feedback, resulting in a debiased recommendation system. In summary, our main contributions are listed as follows:
\begin{itemize}
\item We propose a novel debiasing method, IViDR, for mitigating dual latent confounding biases, i.e., the confounding biases caused by latent confounders between items and user feedback, as well as between item exposure and user feedback.
\item We provide a theoretical analysis of the identification of the learned latent representations in the IViDR, along with a discussion on the validity of using IVs.
\item Extensive experiments on both synthetic and real-world datasets demonstrate that IViDR outperforms state-of-the-art recommendation methods in reducing bias and improving recommendation accuracy.
\end{itemize}

\section{Related Work}
\label{sec:relate} 
In this section, we review methods for recommender systems, which can broadly be categorized into two groups: traditional correlation-based systems and causality-based systems.

\paragraph{Traditional Recommendation Algorithms}
Traditional algorithms include collaborative filtering-based approaches (e.g., Matrix Factorisation (MF) \cite{12}), neural network-based recommendation algorithms (e.g., Deep Factorisation Machine (DeepFM) \cite{deepfm}, Deep Cross Network (DCN) \cite{dcn}, Deep Interest Network (DIN) \cite{din}), and graph neural network-based algorithms (e.g., LightGCN \cite{lightgcn}). However, these methods are susceptible to introducing biases into recommendation outcomes. Instead of delving deeply into traditional algorithms, we direct our attention to causality-based recommender systems. For more comprehensive reviews of traditional approaches, please refer to \cite{r1, r2}.

\paragraph{Causal Recommendation Algorithms for Addressing Specific Biases}
With advancements in causal inference techniques \cite{20,21}, causality-based recommendation algorithms have emerged. Several studies have adopted classical causal inference methods, such as back-door adjustment \cite{20} or IPW \cite{23}, to mitigate specific biases. For instance, IPS \cite{23} uses IPW to address selection bias \cite{15}, while D2Q \cite{34} employs back-door adjustment to reduce video duration bias, which often leads recommender systems to favour longer videos. The Popularity De-biasing Algorithm (PDA) \cite{35} tackles popularity bias to enhance recommendation accuracy. Despite their promise, these algorithms typically overlook the influence of latent confounders, which can undermine the accuracy of classical causal inference methods~\cite{cheng2023causal, cheng2024data}.

To address the challenges posed by latent confounders, recent research has proposed several debiasing methods. The Deconfounder \cite{29} seeks to approximate latent confounders using historical user data to mitigate their effects. Similarly, Hidden Confounder Removal (HCR) \cite{37} employs front-door adjustment to address latent confounders, while iDCF \cite{idcf} utilises proximal causal inference techniques for the same purpose. 
However, these methods may struggle to infer latent confounders directly from interaction data, as some latent confounders may lack proxy variables in data.

Another line of research, exemplified by IV4REC \cite{IV-serach-data}, leverages the pre-defined IVs \cite{iv-4, iv-11, iv-29} to decompose the input vectors of recommendation models and mitigate the effects of latent confounders beween items and user feedback. However, these approaches directly modify the input vectors while overlooking the latent confounders between item exposure and user feedback, leaving more complex types of latent confounding unresolved.

In contrast, our IViDR method address both types of latent confounders in recommender systems by developing a novel IV-based approach combined with an iVAE. Note that our IViDR method effectively handles dual latent confounding biases, with or without proxy variables, resulting in a reliable and unbiased recommender system, as demonstrated in our experiments.

\section{Preliminaries}
In this section, we introduce important notations, definitions, and the concept of IV that are used throughout the paper.


\subsection{Notations}

\begin{table}[t]
  \centering 
  \caption{Definitions and Notations}
  \label{tab:notation}
  \begin{tabular}{c p{5cm}}
    \toprule
    \textbf{Symbol} & \textbf{Definition}\\
    \midrule
    $\mathbf{U}$ & Users. \\
    $\mathbf{I}$ & Items.\\
    $\mathbf{Z}$ & The instrumental variable (the embeddings of user features).\\
    $\mathbf{T}$ & The treatment (the embedding of the target item $i$ and the embeddings of the items interacted with $u$). \\
    $\mathbf{A}$ & Exposure vector.\\
    $\mathbf{R}$ & User feedback. \\
    $\mathbf{W}$ & The set of proxy variables.\\
    $\mathbf{C}$ & Latent confounders with proxy variables. \\
    $\mathbf{B}$ & Latent confounders without proxy variables. \\ 
    $\mathbf{X}$ & The user-item interaction data.\\
    \bottomrule 
  \end{tabular}
\end{table}

The primary notations are summarised in Table~\ref{tab:notation}. 
We represent vectors with bold-faced uppercase (e.g., $\mathbf{X}$), and their elements by lowercase letters with subscripts (e.g., $x_{ij}$).
 
Let $\mathbf{U}$ represent the set of users, where $|\mathbf{U}|=m$ denotes the total number of users. Similarly, let $\mathbf{I}$ represent the set of items, with $|\mathbf{I}|=n$ denoting the total number of items. 
For each user $u \in \mathbf{U}$, the exposure vector $\mathbf{A}$ is defined as $\mathbf{A} = [a_{u1}, a_{u2}, \dots, a_{un}] \in \{0,1\}^n$,  where $a_{ui} = 1/0$ indicates that item $i$ is exposed to user 
$u$, and $a_{ui} = 0$ indicates it is not. The feedback from user $u$ is represented by the vector $\mathbf{R} = [r_{u1}, r_{u2}, \dots, r_{un}]$, where $r_{ui}$ denotes the observed feedback from user $u$ for item $i$.

In this work, we adopt the potential outcomes framework~\cite{imbens2015causal} to develop debiased recommender systems. Let $r_{ui}^a$ denote the potential outcome for user $u$ on item $i$ under the exposure status $a$. Previous studies~\cite{29} have assumed that the exposure of item $i$ to user $u$ is the sole factor influencing this outcome. However, our method accounts for additional covariates or latent variables that may affect the  outcomes $r_{ui}^a$. The set of observed proxy variables for user $u$ is represented by $\mathbf{W}$. The set of user features is indicated by $\mathbf{Z}$. The recommendation dataset $\mathcal{D}$ consists of the interaction data $\mathbf{X}$, the set of user features $\mathbf{Z}$, and the set of proxy variables $\mathbf{W}$. The probability that user $u$ will provide positive feedback on item $i$, given the exposure vector $\mathbf{A}$, is denoted by $p(r_{ui} = 1 \mid \mathbf{A})$.  In recommender systems, the set of treatment variables $\mathbf{T}$ is defined as a set of embeddings, which includes the embedding of the target item $i$ and the embeddings of the items previously interacted with by user $u$, to predict the preference score.

Moreover, we use a graphical causal model~\cite{20, cheng2024data}, specifically a directed acyclic graph (DAG), to represent the causal relationships between variables in recommender systems. Our proposed causal DAG $\mathcal{G}$ is illustrated in Figure~\ref{fig:2}. In the DAG, $\mathcal{G} = (\mathbf{V}, \mathbf{E})$,   $\mathbf{V} = \mathbf{Z} \cup \mathbf{W} \cup \mathbf{C} \cup \mathbf{B} \cup \{\mathbf{T}, \mathbf{A}, \mathbf{R}\}$ is the set of vertices  and $\mathbf{E}$ represents the set of directed edges. Here, $\mathbf{B}$ denotes latent confounders without proxy variables, while $\mathbf{C}$ represents latent confounders with corresponding proxy variables.

\begin{figure}[t]
    \centering
    \includegraphics[width=0.815\linewidth]{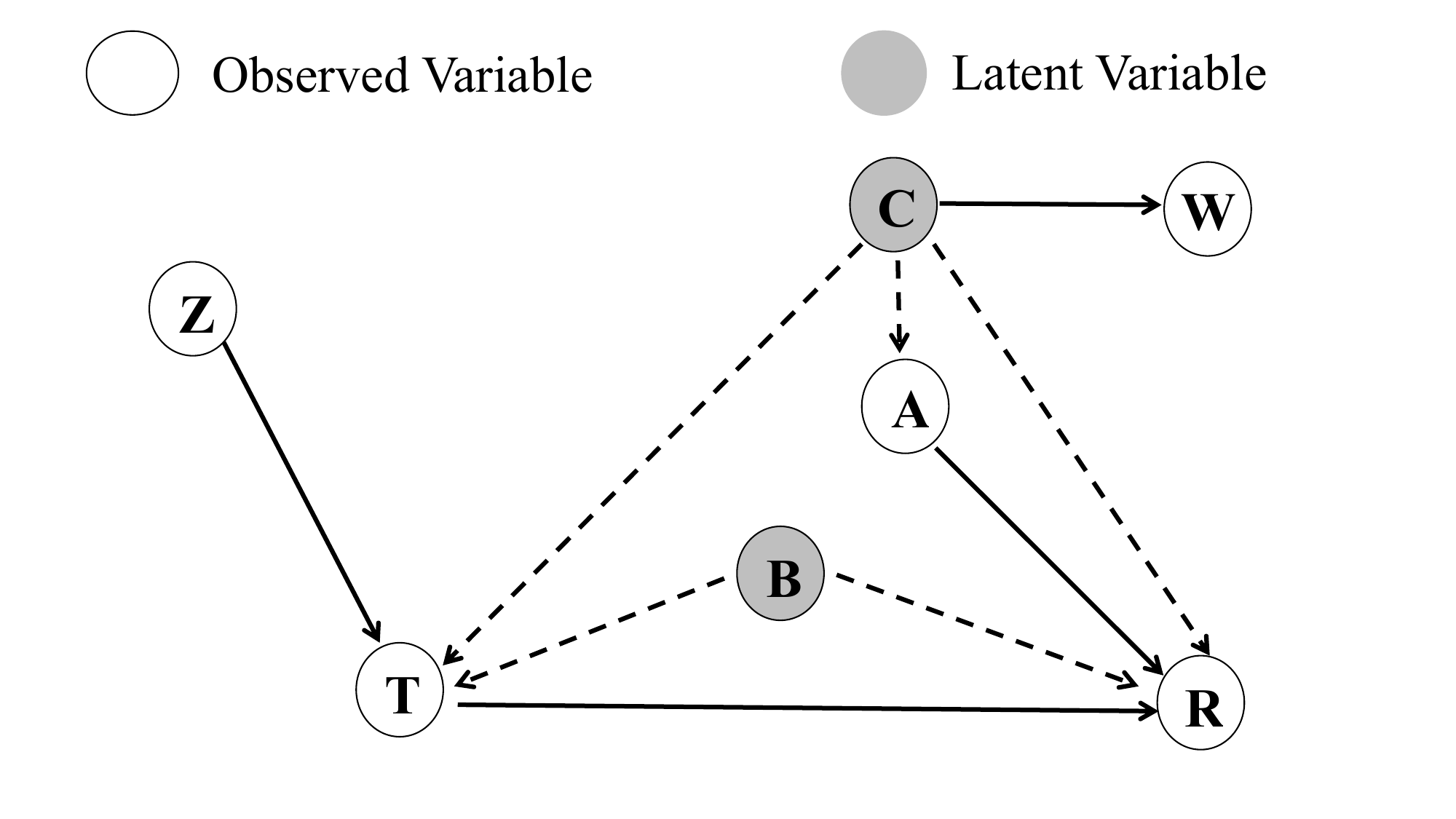}
    \caption{ An example DAG illustrating that $\mathbf{Z}$ serves as an IV. $\mathbf{T}$ and $\mathbf{R}$ represent the treatment and outcome variables, respectively, while $\mathbf{B}$ denotes latent confounders without proxy variables. $\mathbf{C}$ denotes latent confounders with proxy variables, $\mathbf{A}$ indicates exposure status, and $\mathbf{W}$ represents the set of proxy variables.}
    \label{fig:2}
\end{figure}


\subsection{Instrumental Variable (IV) Approach}
The IV approach \cite{iv-4,iv-11,iv-29} is a useful tool for addressing the confounding bias caused by latent confounders. 
The set of variable $\mathbf{Z}$ is a set of valid IVs relative to pair $\mathbf{T}$ and $\mathbf{R}$ when it satisfies the following three conditions:
\begin{definition}[Instrumental Variable (IV)~\cite{20}]
\label{Definition 1}
Given a  DAG $\mathcal{G} = (\mathbf{V}, \mathbf{E})$, where the set of vertices is defined as $\mathbf{V} = \mathbf{Z} \cup \mathbf{W} \cup \mathbf{C} \cup \mathbf{B} \cup \{\mathbf{T}, \mathbf{A}, \mathbf{R}\}$ and the set of directed edges as $\mathbf{E}$, then $\mathbf{Z}$ is a set of valid IVs if the following conditions hold:
\begin{itemize}
    \item[(i)] $\mathbf{Z} \not \nindep_d \mathbf{T}$, meaning there is at least one unblocked path between $\mathbf{Z}$ and $\mathbf{T}$ in $\mathcal{G}$. 
    \item[(ii)] $\mathbf{Z}\indep_d \mathbf{R} \mid \{\mathbf{T}, \mathbf{B}\}$  in  $\mathcal{G}_{\underline{\mathbf{T}}}$ (obtained by removing the edge $\mathbf{T} \rightarrow \mathbf{R}$ from $\mathcal{G}$).
    \item[(iii)] $\mathbf{Z}$ does not share confounders with $\mathbf{R}$, i.e., $(\mathbf{Z} \indep_{d} \mathbf{R})_{\mathcal{G}_{\underline{\mathbf{T}}}}$.
\end{itemize}
\end{definition}

\noindent where $\nindep_d$ and $\indep_d$ refer to d-connection and d-separation, respectively~\cite{20}. Once a variable satisfies the three conditions, it can serve as a valid IV. Existing IV-based methods can be employed to mitigate the confounding bias caused by latent confounders in causal effect estimation. Most IV-based methods utilize the two-stage least squares (2SLS) procedure~\cite{iv-17}. Recently, several IV-based causal learning approaches have extended the 2SLS by incorporating deep learning techniques. For example, DeepIV~\cite{iv-11} provides a flexible framework that combines deep learning with the 2SLS method. DFIV~\cite{iv-36} introduces an alternating training strategy for 2SLS, demonstrating strong performance on high-dimensional data.

Our problem setting is illustrated in the causal DAG shown in Figure \ref{fig:2}, where the set of latent confounders $\mathbf{B}$ exist between $\mathbf{T}$ and  $\mathbf{R}$. In this setting, $\mathbf{Z}$ (the embeddings of user features, such as observed gender) serves as a valid IV. Consequently, we employ IV-based methods~\cite{iv-1,iv-5} to estimate the causal effect of $\mathbf{T}$ on $\mathbf{R}$. Specifically, we use the 2SLS method to reconstruct $\mathbf{T}$ and mitigate the effects of latent confounders $\mathbf{B}$. 

\section{The Proposed IV\lowercase{i}DR Method}
In this section, we provide the details of our proposed IViDR method. First, we present an overview of the IViDR method. Next, we provide a theoretical analysis of the correctness of IVs in the IViDR method. We then introduce the 2SLS method in IViDR to reconstruct treatments and obtain debiased interaction data for $\mathbf{T}$ and $\mathbf{R}$. After that, we describe how latent confounders between item exposure $\mathbf{A}$ and $\mathbf{R}$ are learned using an iVAE within the IViDR method. Additionally, we prove the identifiability of the learned representations. Finally, we present the objective function of the IViDR method.

\begin{figure*}[t]
    \centering{}
    \vspace{-35pt}
    \includegraphics[width=6.15in]{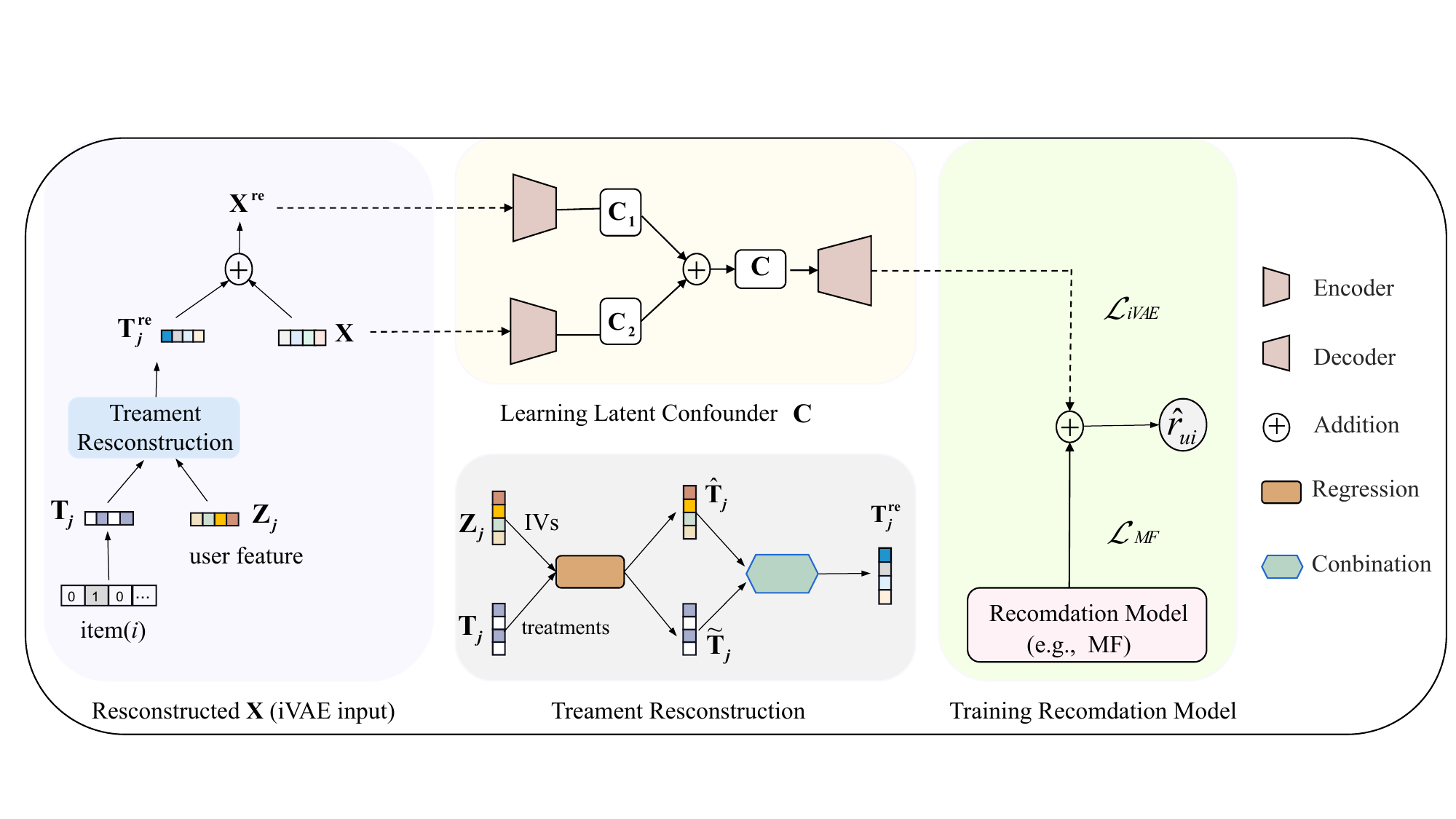} \vspace{-20pt}
    \caption{The architecture of our proposed IViDR method. IViDR uses the embeddings of user features as IVs to decompose treatments into fitted (the values fitted by the regression model) and residual (the residuals) components. It then reconstructs the treatments and combines them with interactional data to generate debiased interactional data. An iVAE is employed to infer latent representations from proxy variables, interactional data, and debiased interactional data. Finally, IViDR adjusts for these latent representations to mitigate confounding biases.
    }
    \label{fig:3}
\vspace{-10pt}
\end{figure*}

\subsection{An Overview for IViDR Method}
The main steps of the IViDR method are illustrated in Figure \ref{fig:3}. In IViDR, we use the embeddings of user features, $\mathbf{Z}$, as instrumental variables (IVs). The details of the IViDR are as follows: First, IViDR uses IVs to decompose the treatments into fitted (which explain why a user prefers an item) and residual components (which capture other factors influencing the user-item interaction). With the aid of IVs, we reconstruct $\mathbf{T}$ to obtain the reconstruct treatment $\mathbf{T}^{\mathrm{re}}$, ensuring that $\mathbf{T}^{\mathrm{re}}$ is not affected by the latent confounders $\mathbf{B}$. The reconstruct treatment $\mathbf{T}^{\mathrm{re}}$ is then combined with the interaction data to generate the debiased interaction data.
Next, IViDR employs the iVAE to efficiently learn identifiable latent representations from a combination of proxy variables, interaction data, and debiased interaction data, mitigating the effects of latent confounders $\mathbf{C}$. Finally, IViDR adjusts for the learned representations to further reduce confounding biases. Additionally, we provide a theoretical analysis of the identifiability of learning the representation $\mathbf{C}$ through our IViDR method. Due to space limitations, the pseudo-code of IViDR is provided in Appendix \ref{appendix:A}.

\subsection{The Soundness of IV in IViDR Method}
In recommender systems, confounding bias typically arises from specific interactions when users access the service. However, since the embeddings of user features represent all users, it is reasonable to conclude that there is no confounding bias between $\mathbf{Z}$ and $\mathbf{R}$. Thus, we utilise the embeddings of user features as a valid IV. Furthermore, we provide a theoretical analysis to demonstrate the soundness of using the embeddings of user features, $\mathbf{Z}$, as a valid IV for mitigating the confounding bias caused by latent confounders:

\begin{theorem}
\label{Theorem 1} 
Given a causal DAG $\mathcal{G} = (\mathbf{Z} \cup \mathbf{W} \cup \mathbf{C} \cup \mathbf{B} \cup \{\mathbf{T}, \mathbf{A}, \mathbf{R}\}, \mathbf{E})$, where $\mathbf{T}$ and $\mathbf{R}$ are the treatment and outcome, respectively, and $\mathbf{E}$ is the set of edges between the variables. Let $\mathbf{B}$ denote the latent confounders between $\mathbf{T}$ and $\mathbf{R}$. Suppose there exists a directed edge $\mathbf{T} \rightarrow \mathbf{R}$ in $\mathbf{E}$, and  $\mathbf{Z}$ represents the embeddings of user features. Additionally, assume that the underlying data generating process is faithful to the proposed causal DAG $\mathcal{G}$ as shown in Figure \ref{fig:2}. Therefore, $\mathbf{Z}$ serves as the valid IV for estimating the causal effect of $\mathbf{T}$ on $\mathbf{R}$.
\end{theorem}

The proof can be found in Appendix~\ref{proof}. Theorem \ref{Theorem 1} allows us to use the embeddings of user features, $\mathbf{Z}$, as valid IVs to decompose the treatment variable, $\mathbf{T}$, for obtaining the reconstructed treatment.

\subsection{Interaction Data Reconstruction}
\label{TRUIV}

We introduce the treatment reconstruction component using the IV (i.e., the embeddings of user features $\mathbf{Z}$) in our IViDR method.

\subsubsection{Construction of treatments and IVs.}

The treatment variable $\mathbf{T}$ in our problem setting is defined as: 

\begin{equation}
\mathbf{T} = \left\{\mathbf{T}_j : j \in \mathbf{I}_u \cup \{i\}\right\}
\end{equation}
where $\mathbf{T}_j \in \mathbb{R}^{d_i}$ is the embedding vector of item $j$, and $\mathbf{I}_u$ represents the set of items interacted with by user $u$ in the recommendation dataset $\mathcal{D}$. $\mathbf{I}_u$ is defined as:

\begin{equation}
\mathbf{I}_u=\left\{i^{\prime}: \exists\left(u, i^{\prime}, r_{ui}\right) \in \mathcal{D}\right\}
\end{equation}

We define a set of matrices $\mathbf{Z}_j$, where each matrix $\mathbf{Z}_j$ represents the embeddings of user features related to item $j$ as follows: 

\begin{equation}
\mathbf{Z}=\left\{\mathbf{Z}_j: j \in \mathbf{I}_u \cup\{i\}\right\}
\end{equation}

Hence, we have the treatment variable $\mathbf{T}$ and the embeddings of user features, $\mathbf{Z}$, for building the debiased recommendation system.

\subsubsection{Treatment reconstruction.}

At this stage, we reconstruct the treatment $\mathbf{T}$ using the IV $\mathbf{Z}$ to obtain the reconstructed treatment $\mathbf{T}^{\mathrm{re}}$, which mitigates the confounding bias caused by latent confounders $\mathbf{B}$ by applying the idea in 2SLS method.
This process can be broken down into two steps: first, treatment decomposition, and second, treatment combination.

\textbf{Step I: Treatment decomposition.}
We regress $\mathbf{T}$ to get the fitted part of the regression $\widehat{\mathbf{T}}$:

\begin{equation}
\widehat{\mathbf{T}}=\left\{\widehat{\mathbf{T}}_j=f_{\text {proj }}\left(\mathbf{T}_j, \mathbf{Z}_j\right): j \in \mathbf{I}_u \cup\{i\}\right\},
\end{equation}
where $\mathbf{T}_j \in \mathbf{T}$ and $\mathbf{Z}_j \in \mathbf{Z}$, and $f_{\text{proj}}: \mathbb{R}^{d_i} \times \mathbb{R}^{d_q \times N} \mapsto \mathbb{R}^{d_q}$ is defined as the product of the matrix $\mathbf{Z}_j$ with an $N$-dimensional vector $\tau_j$.

\begin{equation}
f_{\text {proj }}\left(\mathbf{T}_j, \mathbf{Z}_j\right)=\mathbf{Z}_j \tau_j,
\end{equation}
where $\tau_j$ is the closed-form solution of a least squares regression. $\tau_j$ can be obtained by the following equation: 
\begin{equation}
\tau_j=\underset{\tau_j \in \mathrm{R}^N}{\arg \min }\left\|\mathbf{Z}_j \tau_j-\mathrm{MLP}_0\left(\mathbf{T}_j\right)\right\|_2^2=\mathbf{Z}_j^{\dagger} \mathrm{MLP}_0\left(\mathbf{T}_j\right),
\end{equation}
where $\mathbf{Z}_j^{\dagger}$ is the Moore-Penrose pseudoinverse of $\mathbf{Z}_j$, and $\mathrm{MLP}_0: \mathbb{R}^{d_i} \mapsto \mathbb{R}^{d_q}$ is a multi-layer perceptron. We refer to $\widehat{\mathbf{T}}_j \in \widehat{\mathbf{T}}$ as the fitted part of the embedding $\mathbf{T}_j$.
Then, we obtain the residual part of the regression $\tilde{\mathbf{T}}$ as:
\begin{equation}
\tilde{\mathbf{T}}=\left\{\tilde{\mathbf{T}}_j=\operatorname{MLP}_0\left(\mathbf{T}_j\right)-\widehat{\mathbf{T}}_j: j \in \mathbf{I}_u \cup\{i\}\right\}
\end{equation}

\textbf{Step II: Treatment combination.}
We obtain the reconstructed treatment $\mathbf{T}^{\mathrm{re}}$ by combining $\widehat{\mathbf{T}}$ and $\widetilde{\mathbf{T}}$ as follows:

\begin{equation}
\mathbf{T}^{\mathrm{re}} = \left\{ \mathbf{T}_j^{\mathrm{re}} = \alpha_j^1 \widehat{\mathbf{T}}_j + \alpha_j^2 \tilde{\mathbf{T}}_j : j \in \mathbf{I}_u \cup \{i\} \right\},
\end{equation}
where $\widehat{\mathbf{T}}_j \in \widehat{\mathbf{T}}$ and $\tilde{\mathbf{T}}_j \in \widetilde{\mathbf{T}}$, and $\alpha_j^1 \in \mathbb{R}$ and $\alpha_j^2 \in \mathbb{R}$ are two weights, which are estimated by two multi-layer perceptrons (MLPs):

\begin{equation}
\begin{split}
& \alpha_j^1 = \operatorname{MLP}_1\left(\operatorname{MLP}_0\left(\mathbf{T}_j\right), \mathbf{Z}_j\right), \\
& \alpha_j^2 = \operatorname{MLP}_2\left(\operatorname{MLP}_0\left(\mathbf{T}_j\right), \mathbf{Z}_j\right),
\end{split}
\end{equation}
where the inputs to the two MLPs are concatenations of the transformed $\mathbf{T}_j$ and $\mathbf{Z}_j$, corresponding to item $j$.

\subsubsection{Interaction data reconstruction.}
Finally, we reconstruct the interaction data $\mathbf{X}$ using $\mathbf{T}^{\mathrm{re}}$ to obtain the debiased interaction data $\mathbf{X}^{\mathrm{re}}$, which serves as the input for iVAE in our IViDR method:

\begin{equation}
\mathbf{X}^{\mathrm{re}} = \mathbf{X} + \mathbf{T}_j^{\mathrm{re}}.
\end{equation}

\subsection{Learning the Latent Representation Using iVAE in IViDR Method}
In the previous section, we employed the IV approach to address the confounding bias caused by latent confounders $\mathbf{B}$, but the latent confounders $\mathbf{C}$ still bias the recommender system. To mitigate this, we propose using a proxy variable $\mathbf{W}$ (e.g., the mean price of consumed goods) to recover the representation of latent confounders $\mathbf{C}$.

We assume that $\mathbf{W}$ is a Bernoulli random variable with mean $\mu(\mathbf{C}) \in (0, 1)$ and is correlated with $\mathbf{C}$ given the exposure vector $\mathbf{A}$. The probability $p(r_{ui} = 1 \mid \mathbf{A}, \mathbf{W})$, representing the likelihood that user $u$ gives positive feedback on item $i$ given $\mathbf{A}$ and $\mathbf{W}$, can be inferred from the available dataset.
We further assume that $p(\mathbf{C} \mid \mathbf{A}, \mathbf{W})$, the probability distribution of the estimated latent representation $\mathbf{C}$ given $\mathbf{A}$ and $\mathbf{W}$, is uniquely defined by the factor model~\cite{10.5555/120565.120567}. Next, we show how to identify $p(r_{ui}^a)$ with proxy variables in general, noting that:
\begin{equation}
\label{eq:1}
\begin{split}
p(r_{u i}^a) 
= E_{\mathbf{C}}[p(r_{u i} | \mathbf{A}, \mathbf{C})] 
= \int_c p(\mathbf{C} = c) p(r_{u i} | \mathbf{A}, \mathbf{C} = c) dc.
\end{split}
\end{equation}

From Eq. \cref{eq:1}, it is clear that the key to estimating \( p(r_{ui}^a) \) lies in determining both \( p(\mathbf{C}) \) and \( p(r_{ui} \mid \mathbf{A}, \mathbf{C}) \). 
This process can be broken down into two steps: first, estimating \( p(\mathbf{C}) \), and second, determining \( p(r_{ui} \mid \mathbf{A}, \mathbf{C}) \).

\textbf{Step I:}
The first step is to obtain \( p(\mathbf{C}) \). 
This involves learning \( \mathbf{C} \) from  \( \mathbf{W} \), ensuring that the learned \( \mathbf{C} \) closely approximates the true latent confounders \( \mathbf{C} \) \cite{9,17} as follows:

\begin{equation}
\begin{split}
p(\mathbf{C} = c) &= E_{\mathbf{A}, \mathbf{W}}[p(\mathbf{C} | \mathbf{A}, \mathbf{W})]
\end{split}
\end{equation}

Hence, we need to learn $p\left(\mathbf{C} \mid \mathbf{A}, \mathbf{W}\right)$ from the data. In practice, we use the following equation:

\begin{equation}
\label{eq:3}
\begin{split}
&p(\mathbf{A} | \mathbf{W}) = \int_c p(\mathbf{A} | \mathbf{C} = c) p(\mathbf{C} = c | \mathbf{A}, \mathbf{W}) dc
\end{split}
\end{equation}
where $p(\mathbf{A} | \mathbf{W})$ can be learned directly from the data. Therefore, we need to recover $p\left(\mathbf{C} \mid \mathbf{A}, \mathbf{W}\right)$ and $p\left(\mathbf{A} \mid \mathbf{C}\right)$ from the data. In our IViDR framework, we use iVAE \cite{9} to recover both terms.

\textbf{Step II:}
In the second step, we obtain $p\left(r_{u i} \mid \mathbf{A}, \mathbf{C}\right)$, which can be inferred from $p\left(\mathbf{C} \mid \mathbf{A}, \mathbf{W}\right)$ and $p(r_{u i} | \mathbf{C} = c, \mathbf{A})$ as follows: 
\begin{equation}
\begin{split}
&p(r_{u i} | \mathbf{A}, \mathbf{W}) = \int_c p(r_{u i} | \mathbf{A}, \mathbf{C} = c) p(\mathbf{C} = c | \mathbf{A}, \mathbf{W}) dc
\end{split}
\end{equation}

In IViDR, we also use iVAE  to recover $p(r_{u i} | \mathbf{C} = c, \mathbf{A})$. With these two steps, we can obtain $p\left(\mathbf{C}\right)$ and $p\left(r_{u i} \mid \mathbf{A}, \mathbf{C}\right)$.
By applying Eq.\cref{eq:1}, we can determine $p\left(r_{u i}^a\right)$,  the probability distribution of the potential outcome, from the data.

We use the debiased interactional data $\mathbf{X}^{\mathrm{re}}$ as input to iVAE. To learn $p_\theta({\mathbf{C}}_1 \mid \mathbf{A}, \mathbf{W})$ (the details of obtaining the approximate posterior of latent confounders can be found in Appendix \ref{appendix:C2}), we use $q_\phi({\mathbf{C}}_1 \mid \mathbf{A}, \mathbf{W})$ as the approximate posterior.
We sample the latent confounder ${\mathbf{C}}_1$ from $q_\phi({\mathbf{C}}_1 \mid \mathbf{A}, \mathbf{W})$.
Next, we use the interactional data $\mathbf{X}$ as input to iVAE. To learn $p_\theta({\mathbf{C}}_2 \mid \mathbf{A}, \mathbf{W})$, we use $q_\phi({\mathbf{C}}_2 \mid \mathbf{A}, \mathbf{W})$ as the approximate posterior.
We sample the latent confounder ${\mathbf{C}}_2$ from $q_\phi\left({\mathbf{C}}_2 \mid \mathbf{A}, \mathbf{W}\right)$.
Finally, we fuse ${\mathbf{C}}_1$ and ${\mathbf{C}}_2$ to obtain latent confounder $\mathbf{C}$.

\begin{equation}
\label{eq:15}
\begin{split}
&\mathbf{C} = \rho * {{\mathbf{C}}_1}  +  \tau *{{\mathbf{C}}_2}
\end{split}
\end{equation}
where $\rho$ and  $\tau$ are tuning parameters.

Thus, we can recover the latent confounder \( \mathbf{C} \), and subsequently implement our iVAE submodel within the IViDR method:

\begin{equation}
\mathcal{L}_{iVAE} = E_{\rho * {q_\phi\left({\mathbf{C}}_1 \mid \mathbf{A}, \mathbf{W}\right)} + \tau * {q_\phi\left({\mathbf{C}}_2 \mid \mathbf{A}, \mathbf{W}\right)} }
\end{equation}

\subsection{The Identifiability of the Learned Representation}
In this section, we prove that the estimated latent representation $\mathbf{C}$ learned from $\mathbf{X}^{\mathrm{re}}$ is identifiable.
The identifiability of $\mathbf{C}$ learned by our IViDR is proved as follows. 

Let $\theta=(f, H, \lambda)$ be the parameter of the model $\Theta$ that generates the following condition:
\begin{equation}
\begin{split}
\label{zm1}
p_\theta\left(\mathbf{X}^{\mathrm{re}}, \mathbf{C} \mid \mathbf{W}\right)=p_f\left(\mathbf{X}^{\mathrm{re}} \mid \mathbf{C}\right) p_{H, \lambda}(\mathbf{C} \mid \mathbf{W})
\end{split}
\end{equation}
where \( p_f\left(\mathbf{X}^{\mathrm{re}} \mid \mathbf{C}\right) \) is defined as:

\begin{equation}
\begin{split}
\label{zm2}
p_f\left(\mathbf{X}^{\mathrm{re}} \mid \mathbf{C}\right)=p_{\varepsilon}\left(\mathbf{X}^{\mathrm{re}}-\mathbf{f}(\mathbf{C})\right)
\end{split}
\end{equation}
where the value of \( \mathbf{X}^{\mathrm{re}} \) is decomposed as \( \mathbf{X}^{\mathrm{re}} = f(\mathbf{C}) + \varepsilon \). Here, \( \varepsilon \) is an independent noise variable (independent of both \( \mathbf{C} \) and \( f \)), and its probability density function is given by \( p_{\varepsilon}(\varepsilon) \).

As stated in Assumption \ref{assumption 1}, we assume that the conditional distribution \( p_{H, \lambda}(\mathbf{C} \mid \mathbf{W}) \) is a conditional factor distribution and follows an exponential family of distributions, illustrating the relationship between the latent representation \( \mathbf{C} \) and the proxy variables \( \mathbf{W} \).

\begin{assumption}
\label{assumption 1}
The conditioning on $\mathbf{W}$ utilizes an arbitrary function—like a look-up table or a neural network—that produces the specific exponential family parameters $\lambda_{i, j}$. Consequently, the probability density function is expressed as follows:
\begin{equation}
\begin{split}
\label{zm3}
& p_{H, \lambda}(\mathbf{C} \mid \mathbf{W})= \prod_i \frac{Q_i\left(\mathbf{C}_{i}\right)}{Z_i(\mathbf{W})} \exp \left[\sum_{j=1}^k H_{i, j}\left(\mathbf{C}_{i}\right) \lambda_{i, j}(\mathbf{W})\right]
\end{split}
\end{equation}
where $Q_i$ is the base measure, $Z_i(\mathbf{W})$ is the normalization constant, $\lambda_i(\mathbf{W})=\left(\lambda_{i, 1}(\mathbf{W}), \ldots, \lambda_{i, k}(\mathbf{W})\right)$ is a parameter related to $\mathbf{W}$, $\mathbf{H}_i=\left(H_{i, 1}, \ldots, H_{i, k}\right)$ is the sufficient statistic, and $k$ is the dimension of each sufficient statistic, held constant. 
\end{assumption}

To prove the identifiability of the learned representation $\mathbf{C}$ in our IViDR, we introduce the following definitions \cite{9}:

\begin{definition} [Identifiability classes]
\label{Definition 2}
Let $\sim$ be an equivalence relation on $\Theta$. We say that Eq. \cref{zm1} is identifiable (or $\sim$ identifiable) under $\sim$ if:
\begin{equation}
\begin{split}
p_{\mathbf{\theta}}(\mathbf{X}^{\mathrm{re}})=p_{\tilde{\theta}}(\mathbf{X}^{\mathrm{re}}) \Longrightarrow \tilde{\mathbf{\theta}} \sim \mathbf{\theta}
\end{split}
\end{equation}

$\Theta / \sim$ are called identifiability classes.
\end{definition}

We now define the equivalence relation on the parameter set \( \Theta \).
\begin{definition}
Let $\sim_M$ be the equivalence relation defined on $\Theta$ as follows:

\begin{equation}
\begin{split}
& (\mathbf{f}, \mathbf{H}, \mathbf{\lambda}) \sim(\tilde{\mathbf{f}}, \tilde{\mathbf{H}}, \tilde{\mathbf{\lambda}}) 
 \Leftrightarrow \exists M, \mathbf{q} \mid \mathbf{H}\left(\mathbf{f}^{-1}(\mathbf{X}^{\mathrm{re}})\right)
\\& =M \tilde{\mathbf{H}}\left(\tilde{\mathbf{f}}^{-1}(\mathbf{X}^{\mathrm{re}})\right)+\mathbf{q}, \forall \mathbf{X}^{\mathrm{re}} \in \mathcal{X}
\end{split}
\end{equation}
where $M$ is a $n k \times n k$ matrix and $\mathbf{q}$ is a vector.

\end{definition}

The identifiability of the learned representation 
$\mathbf{C}$ in our IViDR can be derived as follows:

\begin{theorem}
\label{Theorem 2}
Assuming we have data collected from a generative model as defined by Eqs. \cref{zm1}-\cref{zm3}, with parameters \( (\mathbf{f}, \mathbf{H}, \mathbf{\lambda}) \). Suppose the following conditions hold:
\begin{enumerate}[label=(\roman *)]
\item 
\label{i}
The set $\left\{\mathbf{X}^{\mathrm{re}} \in \mathcal{X} \mid \varphi_{\varepsilon}(\mathbf{X}^{\mathrm{re}})=0\right\}$ has measure zero, where $\varphi_{\varepsilon}$ is the  characteristic function of the density $p_{\varepsilon}$ defined in Eq.\cref{zm2}. 
\item 
\label{ii}
The hybrid function $\mathbf{f}$ in Eq.\cref{zm2} is injective.
\item
\label{iii}
The sufficient statistic $H_{i, j}$ in Eq.\cref{zm3} is differentiable almost everywhere, and $\left(H_{i, j}\right)_{1 \leq j \leq k}$ is linearly uncorrelated on any subset of $\mathcal{X}$ that has a measure greater than zero.
\item
\label{iv}
There exist $n k+1$ distinct points $\mathbf{W}_0, \ldots, \mathbf{W}_{n k}$ such that the matrix
\begin{equation}
\begin{split}
L=\left(\mathbf{\lambda}\left(\mathbf{W}_1\right)-\mathbf{\lambda}\left(\mathbf{W}_0\right), \ldots, \mathbf{\lambda}\left(\mathbf{W}_{n k}\right)-\mathbf{\lambda}\left(\mathbf{W}_0\right)\right)
\end{split}
\end{equation}
of size $n k \times n k$ is invertible.
\end{enumerate}
Then the parameters $(\mathbf{f}, \mathbf{H}, \mathbf{\lambda})$ are $\sim_M$ identifiable.
\end{theorem}

The proof is structured into three main steps.

In the first step, we apply the simple convolution technique, as permitted by Assumption \hyperref[i]{(i)}, to transform the equation for the debiased interaction data distribution into one representing the noise-free distribution. This effectively reduces the noisy scenario to a noise-free case, leading to Eq. \cref{zm12}.

In the second step, we eliminate all terms related to the debiased interaction data \( \mathbf{X}^{\mathrm{re}} \) and the proxy variable \( \mathbf{W} \). This is accomplished by leveraging  the conditions provided by assumption \hyperref[iv]{(iv)} and using \( \mathbf{W}_0 \) as the pivot. This is detailed in Eqs. \cref{zm12}-\cref{zm17}.

Finally, we demonstrate that the linear transformation is invertible, establishing an equivalence relation. This final step relies on Assumption \hyperref[iii]{(iii)}.

Detailed proofs of the three steps can be found in Appendix \ref{appendix:B}.

\subsection{The Proposed IViDR Framework}
We propose a generalized framework that is compatible with any recommendation model. For comparison, we select the Matrix Factorisation (MF) model as the recommendation model for our framework. The predicted rating \( \mathcal{L}_{MF} \) is given by:

\begin{equation}
\begin{split}
\mathcal{L}_{MF} =p_u^T q_i + k_u + k_i
\end{split}
\end{equation} 
where \( p_u \) is the latent vector for user \( u \), \( q_i \) is the latent vector for item \( i \), and \( T \) denotes the transpose operation. Furthermore, \( k_u \) represents the user preference bias term, and \( k_i \) is the item preference bias.

Finally, we utilise a point-wise recommendation model parameterised by \( \eta \) to estimate \( p\left(r_{ui} \mid \mathbf{A}, \mathbf{C}\right) \). In particular, we utilize a additive model \( f(u, i, \mathbf{C}; \eta) \). The additive model is:

\begin{equation}
\begin{split}
&f(u, i, \mathbf{C} ; \eta) = \phi *\mathcal{L}_{iVAE}  +  \lambda*\mathcal{L}_{MF}
\end{split}
\end{equation} where $\phi, \lambda$ are tuning parameters. Therefore, our final loss function for IViDR is defined as:

\begin{equation}
\label{eq:18}
\small
\begin{split}
\mathcal{L}_{IViDR} &= \frac{1}{|\mathcal{D}|} \sum_{(u, i) \in \mathcal{D}}  l\left(E_{\rho * {q_\phi\left({\mathbf{C}}_1 \mid \mathbf{A}, \mathbf{W}\right)} +  \tau * {q_\phi\left({\mathbf{C}}_2 \mid \mathbf{A}, \mathbf{W}\right)}}  \left[f(u, i, \mathbf{C} ; \eta)\right], r_{u i}\right)
\end{split}
\end{equation}
where \( l(\cdot, \cdot) \) denotes the binary cross-entropy (BCE) loss, which is commonly used in such models.

\paragraph{Limitations} The success of IViDR relies on the assumptions of IVs and the existence of proxy variables in the data. However, the relationship between user preferences, features, and recommendation mechanisms is often complex, making it challenging to validate the assumptions and ensure that appropriate IVs and proxy variables are available. Additionally, the success of IViDR depends on the quality and relevance of these variables, which may not always be guaranteed in real-world datasets.

\begin{table}
\small
\caption{The statistic of Coat, Yahoo!R3, and KuaiRand.}

    \centering
\begin{tabular}{ p{1.4cm} p{1cm} p{1cm} p{1.2cm} p{1.2cm}  } 
 \hline
 Dataset & \#User & \#Item & \#Biased Data & \#Unbiased Data\\ 
 \hline
 Coat & 290 & 300 & 6,960 & 4,640\\ 
 Yahoo! R3 & 5,400 & 1,000 & 129,179 & 54,000\\ 
 KuaiRand  & 23,533 & 6,712 & 1,413,574 & 954,814 \\ 
 \hline
\end{tabular}
\label{table: statistic}
\vspace{-10pt}
\end{table}

\begin{table*}[t]
\caption{The performance metrics for all methods across three real-world datasets are provided in detail. The best method is highlighted in bold, and the second-best is underlined. We report the p-values from t-tests to assess the statistical significance of the performance differences.}
\centering
\resizebox{\textwidth}{!}{
\label{tab:2} 
\begin{tabular}{c|cc|cc|cc}
\hline
\multirow{2}{*}{Datasets} & \multicolumn{2}{|c|}{Coat} & \multicolumn{2}{c|}{Yahoo!R3} & \multicolumn{2}{c}{KuaiRand} \\
& NDCG@5 & RECALL@5 & NDCG@5 & RECALL@5 & NDCG@5 & RECALL@5 \\
\hline 
MF & $0.5590 \pm 0.0126$ & $0.5418 \pm 0.0116$ & $0.5624 \pm 0.0087$ & $0.7125 \pm 0.0076$ & $0.3732 \pm 0.0011$ & $0.3244 \pm 0.0008$ \\
MF-WF & $0.5577 \pm 0.0120$ & $0.5336 \pm 0.0102$ & $0.5617 \pm 0.0087 $ & $0.7120 \pm 0.0127$ & $0.3710 \pm 0.0006$ & $0.3238 \pm 0.0009$ \\
IPS & $0.5511 \pm 0.0154$ & $0.5358 \pm 0.0113$ & $0.5500 \pm 0.0037$ & $0.6961 \pm 0.0063$ & $0.3698 \pm 0.0014$ & $0.3234 \pm 0.0011$ \\
RD-IPS & $0.5484 \pm 0.0144$ & $0.5302 \pm 0.0176$ & $0.5334 \pm 0.0116$ & $0.6796 \pm 0.0146$ & $0.3469 \pm 0.0009$ & $0.3116  \pm 0.0012$ \\
InvPref & $0.5372 \pm 0.0110$ & $0.5247 \pm 0.0097$ & $0.5881 \pm 0.0037$ & $0.7388 \pm 0.0044$ & $0.3797 \pm 0.0010$ & $0.3282 \pm 0.0005$ \\
DDCF-MF & $0.5631 \pm 0.0068$ & $0.5425 \pm 0.0107$ & $0.6260 \pm 0.0064$ & $0.7636 \pm 0.0052$ & $0.3959 \pm 0.0007$ & $0.3456 \pm 0.0009$ \\
IV4R-MF & $0.5617 \pm 0.0099$ & $0.5433 \pm 0.0097$ & $0.5653 \pm 0.0062$ & $0.7152 \pm 0.0085$ & $0.3741 \pm 0.0010$ & $0.3251 \pm 0.0008$ \\
iDCF  & $\underline{0.5638} \pm 0.0124$ & $\underline{0.5493} \pm 0.0111$ & $\underline{0.6410} \pm 0.0022$ & $\underline{0.7780} \pm 0.0039$ & $\underline{0.4080} \pm 0.0004$ & $\underline{0.3481} \pm 0.0008$ \\
IViDR  & $\mathbf{0.5903} \pm 0.0101$ & $\mathbf{0.5783} \pm 0.0114$ & $\mathbf{0.6602} \pm 0.0038$ & $\mathbf{0.7901} \pm 0.0037$ & $\mathbf{0.4161} \pm 0.0004$ & $\mathbf{0.3549} \pm 0.0006$ \\ 
\hline
p-value & $6 e^{-5}$ & $1 e^{-5}$ & $1 e^{-9}$ & $1 e^{-6}$ & $5 e^{-20}$ & $1 e^{-13}$ \\
\hline
\end{tabular}
} 
\end{table*}

\section{Experiments}
\label{sec:experiment}

In this section, we perform extensive experiments to evaluate the effectiveness of our IViDR method within recommender systems, especially when latent confounders are present.

\subsection{Experimental Settings}
 We begin by introducing the datasets, followed by descriptions of the baseline methods and evaluation metrics. Lastly, we present the parameter settings used for the IViDR model.

\paragraph{Dataset description.}

We conduct experiments on three real-world datasets, summarized in Table \ref{table: statistic}. Details of the datasets can be found in Appendix \ref{appendix:D1}.

\paragraph{Baselines.}
We compare our method with several state-of-the-art (SOTA) deconfounding methods. 
(1) MF \cite{12} \& MF with features (MF-WF): MF decomposes a rating matrix into the product of two lower dimensional matrices. 
MF-WF is an improved version of MF.
(2) IPS \cite{23} \& RD-IPS \cite{4}:  IPS estimates causal effects in observations, and reduces selection bias by IPW. RD-IPS is an improved version of IPS. 
(3) InvPref \cite{30}: InvPref decouples a user's true preferences from biased behaviour and mitigates bias using invariant learning.
(4) DDCF-MF \cite{38}: DeepDCF combines deep learning and collaborative filtering techniques to predict user preferences for items. DDCF-MF uses MF as the backbone.
(5) IV4R-MF \cite{IV-serach-data}: IV4Rec applies IVs to decompose input vectors within recommendation models to mitigate the effects of latent confounders. IV4R-MF uses MF as the backbone.
(6) iDCF \cite{idcf}: The iDCF incorporates proximal causal inference techniques to identifiably predict users' counterfactual feedback on items.

\paragraph{Evaluation metrics.}
We evaluate the top five recommendations using RECALL@5 and NDCG@5 metrics. Each method is tested across ten trials, and we report the average and standard deviation to demonstrate the effectiveness of each method.

\paragraph{Implementation and Configuration.}
We implement our IViDR model in PyTorch. A grid search is performed to determine the optimal hyperparameters for each method. The learning rate is selected from the set \{1e-3, 5e-4, 1e-4, 5e-5, 1e-5\}, and the weight decay is chosen from \{1e-5, 1e-6\}. To ensure a fair comparison, we utilize the \textit{Adam} optimizer~\cite{kingma2014adam} for all methods, including our IViDR method.

We set the parameters \( \phi \) and \( \lambda \) in Eq. \cref{eq:18} to 1, and \( \rho \) and \( \tau \) in Eq. \cref{eq:15} to 0.9. The treatment embedding and user feature embeddings are 32 dimensions in the Coat dataset, 96 dimensions in the Yahoo!R3 dataset, and 128 dimensions in the KuaiRand dataset.

\begin{table*}[t]
\caption{The ablation study of our IViDR method on three real-world datasets. }
\centering
\resizebox{\textwidth}{!}{
\label{tab:as} 
\begin{tabular}{c|cc|cc|cc}
\hline
\multirow{2}{*}{Datasets} & \multicolumn{2}{|c|}{Coat} & \multicolumn{2}{c|}{Yahoo!R3} & \multicolumn{2}{c}{KuaiRand} \\
& NDCG@5 & RECALL@5 & NDCG@5 & RECALL@5 & NDCG@5 & RECALL@5 \\
\hline 
IV4R-MF & $0.5617 \pm 0.0099$ & $0.5433 \pm 0.0097$ & $0.5653 \pm 0.0062$ & $0.7152 \pm 0.0085$ & $0.3741 \pm 0.0010$ & $0.3251 \pm 0.0008$ \\
iDCF  & ${0.5638} \pm 0.0124$ & ${0.5493} \pm 0.0111$ & ${0.6410} \pm 0.0022$ & ${0.7780} \pm 0.0039$ & ${0.4080} \pm 0.0004$ & ${0.3481} \pm 0.0008$ \\
IViDR-T  & ${0.5669} \pm 0.0103$ & ${0.5513} \pm 0.0106$ & ${0.6428} \pm 0.0039$ & ${0.7794} \pm 0.0036$ & ${0.4104} \pm 0.0006$ & ${0.3501} \pm 0.0007$ \\ 
IViDR-F  & ${0.5846} \pm 0.0121$ & ${0.5710} \pm 0.0094$ & $\underline{0.6552} \pm 0.0041$ & ${0.7849} \pm 0.0045$ & ${0.4130} \pm 0.0005$ & ${0.3529} \pm 0.0007$ \\ 
IViDR-R  & $\underline{0.5875} \pm 0.0105$ & $\underline{0.5728} \pm 0.0101$ & ${0.6514} \pm 0.0035$ & $\underline{0.7857} \pm 0.0030$ & $\underline{0.4151} \pm 0.0005$ & $\underline{0.3541} \pm 0.0008$ \\ 

IViDR  & $\mathbf{0.5903} \pm 0.0101$ & $\mathbf{0.5783} \pm 0.0114$ & $\mathbf{0.6602} \pm 0.0038$ & $\mathbf{0.7901} \pm 0.0037$ & $\mathbf{0.4161} \pm 0.0004$ & $\mathbf{0.3549} \pm 0.0006$ \\ \hline

\end{tabular}
} 
\end{table*}

\subsection{Performance Comparison}
Table \ref{tab:2} presents the experimental results for each method, which are also visualized in Figure \ref{fig:4} (with Figure \ref{fig:4} located in Appendix \ref{appendix:E1}). The results demonstrate that IViDR significantly outperforms all other algorithms across all datasets and evaluation metrics. From Table \ref{tab:2} and Figure \ref{fig:4}, we observe the following:  
(1) Our proposed IViDR achieves the best performance across all datasets and metrics. The p-values indicate that the performance improvement of IViDR over other methods is statistically significant. This demonstrates that IViDR effectively enhances recommendation accuracy and reduces bias.  
(2) iDCF across all datasets indicates its effectiveness in mitigating confounding bias through the use of latent confounders. 
(3) IV4R-MF exhibits clear improvements over the basic MF method. By leveraging instrumental variables, IV4R-MF enhances the accuracy of recommendations across multiple datasets.  
(4) DDCF-MF shows notable improvements over the basic MF and MF-WF methods, especially on the Yahoo!R3 and KuaiRand datasets. The incorporation of deep learning components allows DDCF-MF to effectively capture complex feature interactions, enhancing recommendation accuracy and robustness.


\subsection{Ablation Study}

To further understand the impact of each component in our model, we performed ablation studies. We verified the following cases: 
(1) The case where the original treatment \( \mathbf{T} \) is added to the interactional data (IViDR-T). 
(2) The case where only the fitted part \( \widehat{\mathbf{T}} \) is added to the interactional data (IViDR-F). 
(3) The case where only the residual part \( \widetilde{\mathbf{T}} \) is added to the interactional data (IViDR-R). 

The experimental results are shown in Table \ref{tab:as}. A detailed analysis of the ablation studies can be found in Appendix \ref{appendix:E1}. We also performed a hyperparameter analysis, with a detailed discussion and complete results provided in Appendix \ref{appendix:E2}.

\subsection{The Correctness of the Learned Latent Representation Using Our IViDR}

We optimised the simulation data generation method from iDCF \cite{idcf} to better match realistic scenarios. Specifically, we changed the economic level of users from a uniform distribution to a Gaussian-like distribution. We increased the number of users from 2,000 to 10,000, and the number of items from 300 to 1,000. Additionally, we added the mixing of item factors (see Appendix \ref{appendix:F}).

There are three important hyperparameters: \( \alpha \), which controls the density of the exposure vector; \( \beta \), which represents the weight of the confounding effect on user preferences; and \( \gamma \), which controls the weight of random noise in the user’s exposure.

\begin{figure}[t]
    \centering
    \begin{minipage}[b]{0.31\linewidth}
        \centering
        \includegraphics[width=\linewidth]{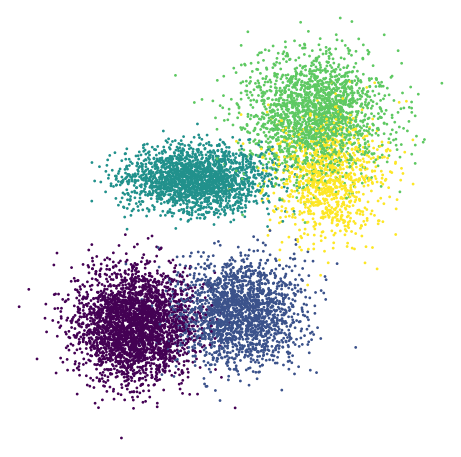}
         \captionsetup{labelformat=empty} 
        \caption*{(a) Ground truth.}
        \label{t}  
    \end{minipage}
    \hfill
    \begin{minipage}[b]{0.31\linewidth}
        \centering
        \includegraphics[width=\linewidth]{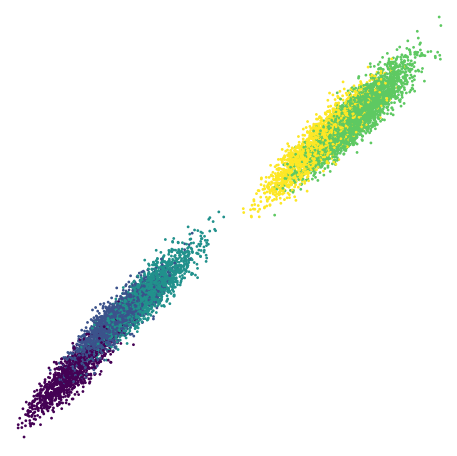}
         \captionsetup{labelformat=empty} 
        \caption*{(b) iDCF}
    \end{minipage}
    \hfill
    \begin{minipage}[b]{0.31\linewidth}
        \centering
        \includegraphics[width=\linewidth]{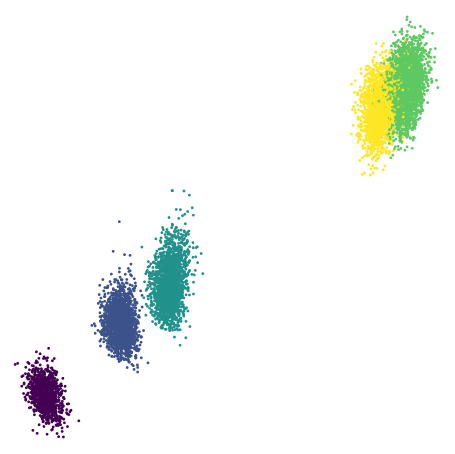}
         \captionsetup{labelformat=empty} 
        \caption*{(c) IViDR(ours).}
    \end{minipage}
    \caption{(a) The visualisation of the true latent confounder. (b) The estimated latent confounders using iDCF. (c) The estimated latent confounders using IViDR.}
    \label{fig123}
    \vspace{-10pt}
\end{figure}

\begin{table}[t]
\small
\caption{Comparison of Models}
\centering
\begin{tabular}{p{1.6cm}| p{0.8cm}| p{0.8cm} | p{0.8cm} | p{0.8cm} | p{0.8cm} } 
\hline
\diagbox[width=2cm,height=1.1cm]{Model}{$\gamma$} & $0.0$ & $5.0$ & $10.0$ & $15.0$ & $20.0$ \\
\hline
iDCF & $0.8162$ & $0.7034$ & $0.6475$ & $0.5449$ & $0.4264$ \\
IViDR(ours)  & $0.8405$ & $0.7826$ & $0.7659$ & $0.6879$ & $0.6262$ \\
\hline
\end{tabular}
\label{tab:tr}
\vspace{-5pt}
\end{table}

\textbf{The visualisation of the latent confounder.}
Figure \ref{fig123}a shows the ground truth of the latent confounder with the exposure noise weight \( \gamma = 0 \). 
Figure \ref{fig123}b shows that the latent confounders predicted by iDCF still exhibit some overlap. 
Figure \ref{fig123}c demonstrates that the structure of the clusters of latent confounders estimated by IViDR is much closer to the true latent confounder distribution.

\textbf{Effect of Learning Identifiable Latent Confounders.}
To accurately evaluate the difference between the estimated and true latent confounders, we calculated the mean correlation coefficient (MCC) between them. 
MCC is a widely used metric for evaluating the accuracy of estimated latent confounders \cite{9}. 
For comparison, we adopted the same exposure noise settings (fixed \( \alpha = 0.1 \) and \( \beta = 2.0 \)) as in iDCF \cite{idcf}. 
The results, shown in Table \ref{tab:tr}, indicate that IViDR achieves more accurate estimations than iDCF, particularly as the exposure noise weight \( \gamma \) increases.

\vspace{-15pt}

\section{Conclusion}
In this paper, we propose a novel debiasing method (IViDR) that jointly integrates the  IV method and iVAE for mitigating  dual latent confounding biases in recommender systems. By leveraging the IV method and the iVAE, IViDR effectively mitigates dual confounding biases caused by latent confounders both between items and user feedback, and between item exposure and user feedback. This dual capability makes IViDR more robust compared to existing methods, which typically only address a single type of latent confounder. Our IViDR method uses user feature embeddings as IVs to reconstruct treatments and generate debiased interaction data, thereby mitigating confounding between item embeddings and user feedback. IViDR then uses iVAE to infer identifiable representations of latent confounders affecting item exposure and user feedback. We provided theoretical analyses demonstrating the soundness of using IVs and the identifiability of latent representations. Experimental evaluations on synthetic and three real-world datasets demonstrate that IViDR accurately estimates latent confounders and significantly reduces bias compared to existing methods, leading to improved recommendation performance.

\bibliographystyle{ACM-Reference-Format}
\balance
\bibliography{sample-base}


\begin{thebibliography}{44}


\ifx \showCODEN    \undefined \def \showCODEN     #1{\unskip}     \fi
\ifx \showDOI      \undefined \def \showDOI       #1{#1}\fi
\ifx \showISBNx    \undefined \def \showISBNx     #1{\unskip}     \fi
\ifx \showISBNxiii \undefined \def \showISBNxiii  #1{\unskip}     \fi
\ifx \showISSN     \undefined \def \showISSN      #1{\unskip}     \fi
\ifx \showLCCN     \undefined \def \showLCCN      #1{\unskip}     \fi
\ifx \shownote     \undefined \def \shownote      #1{#1}          \fi
\ifx \showarticletitle \undefined \def \showarticletitle #1{#1}   \fi
\ifx \showURL      \undefined \def \showURL       {\relax}        \fi
\providecommand\bibfield[2]{#2}
\providecommand\bibinfo[2]{#2}
\providecommand\natexlab[1]{#1}
\providecommand\showeprint[2][]{arXiv:#2}

\bibitem[Abdollahpouri et~al\mbox{.}(2019)]%
        {popularity1}
\bibfield{author}{\bibinfo{person}{Himan Abdollahpouri}, \bibinfo{person}{Masoud Mansoury}, \bibinfo{person}{Robin Burke}, {and} \bibinfo{person}{Bamshad Mobasher}.} \bibinfo{year}{2019}\natexlab{}.
\newblock \showarticletitle{The unfairness of popularity bias in recommendation}.
\newblock \bibinfo{journal}{\emph{arXiv preprint arXiv:1907.13286}} (\bibinfo{year}{2019}).
\newblock


\bibitem[Adomavicius and Tuzhilin(2005)]%
        {r1}
\bibfield{author}{\bibinfo{person}{Gediminas Adomavicius} {and} \bibinfo{person}{Alexander Tuzhilin}.} \bibinfo{year}{2005}\natexlab{}.
\newblock \showarticletitle{Toward the next generation of recommender systems: A survey of the state-of-the-art and possible extensions}.
\newblock \bibinfo{journal}{\emph{IEEE transactions on knowledge and data engineering}} \bibinfo{volume}{17}, \bibinfo{number}{6} (\bibinfo{year}{2005}), \bibinfo{pages}{734--749}.
\newblock


\bibitem[Belloni et~al\mbox{.}(2012)]%
        {iv-1}
\bibfield{author}{\bibinfo{person}{Alexandre Belloni}, \bibinfo{person}{Daniel Chen}, \bibinfo{person}{Victor Chernozhukov}, {and} \bibinfo{person}{Christian Hansen}.} \bibinfo{year}{2012}\natexlab{}.
\newblock \showarticletitle{Sparse models and methods for optimal instruments with an application to eminent domain}.
\newblock \bibinfo{journal}{\emph{Econometrica}} \bibinfo{volume}{80}, \bibinfo{number}{6} (\bibinfo{year}{2012}), \bibinfo{pages}{2369--2429}.
\newblock


\bibitem[Caner and Hansen(2004)]%
        {iv-4}
\bibfield{author}{\bibinfo{person}{Mehmet Caner} {and} \bibinfo{person}{Bruce~E Hansen}.} \bibinfo{year}{2004}\natexlab{}.
\newblock \showarticletitle{Instrumental variable estimation of a threshold model}.
\newblock \bibinfo{journal}{\emph{Econometric theory}} \bibinfo{volume}{20}, \bibinfo{number}{5} (\bibinfo{year}{2004}), \bibinfo{pages}{813--843}.
\newblock


\bibitem[Cheng et~al\mbox{.}(2024)]%
        {cheng2024data}
\bibfield{author}{\bibinfo{person}{Debo Cheng}, \bibinfo{person}{Jiuyong Li}, \bibinfo{person}{Lin Liu}, \bibinfo{person}{Jixue Liu}, {and} \bibinfo{person}{Thuc~Duy Le}.} \bibinfo{year}{2024}\natexlab{}.
\newblock \showarticletitle{Data-driven causal effect estimation based on graphical causal modelling: A survey}.
\newblock \bibinfo{journal}{\emph{Comput. Surveys}} \bibinfo{volume}{56}, \bibinfo{number}{5} (\bibinfo{year}{2024}), \bibinfo{pages}{1--37}.
\newblock


\bibitem[Cheng et~al\mbox{.}(2022)]%
        {cheng2022local}
\bibfield{author}{\bibinfo{person}{Debo Cheng}, \bibinfo{person}{Jiuyong Li}, \bibinfo{person}{Lin Liu}, \bibinfo{person}{Jiji Zhang}, \bibinfo{person}{Jixue Liu}, {and} \bibinfo{person}{Thuc~Duy Le}.} \bibinfo{year}{2022}\natexlab{}.
\newblock \showarticletitle{Local search for efficient causal effect estimation}.
\newblock \bibinfo{journal}{\emph{IEEE Transactions on Knowledge and Data Engineering}} \bibinfo{volume}{35}, \bibinfo{number}{9} (\bibinfo{year}{2022}), \bibinfo{pages}{8823--8837}.
\newblock


\bibitem[Cheng et~al\mbox{.}(2023)]%
        {cheng2023causal}
\bibfield{author}{\bibinfo{person}{Debo Cheng}, \bibinfo{person}{Ziqi Xu}, \bibinfo{person}{Jiuyong Li}, \bibinfo{person}{Lin Liu}, \bibinfo{person}{Jixue Liu}, {and} \bibinfo{person}{Thuc~Duy Le}.} \bibinfo{year}{2023}\natexlab{}.
\newblock \showarticletitle{Causal inference with conditional instruments using deep generative models}. In \bibinfo{booktitle}{\emph{Proceedings of the AAAI conference on artificial intelligence}}, Vol.~\bibinfo{volume}{37}. \bibinfo{pages}{7122--7130}.
\newblock


\bibitem[Chernozhukov et~al\mbox{.}(2007)]%
        {iv-5}
\bibfield{author}{\bibinfo{person}{Victor Chernozhukov}, \bibinfo{person}{Guido~W Imbens}, {and} \bibinfo{person}{Whitney~K Newey}.} \bibinfo{year}{2007}\natexlab{}.
\newblock \showarticletitle{Instrumental variable estimation of nonseparable models}.
\newblock \bibinfo{journal}{\emph{Journal of Econometrics}} \bibinfo{volume}{139}, \bibinfo{number}{1} (\bibinfo{year}{2007}), \bibinfo{pages}{4--14}.
\newblock


\bibitem[Ding et~al\mbox{.}(2022)]%
        {4}
\bibfield{author}{\bibinfo{person}{Sihao Ding}, \bibinfo{person}{Peng Wu}, \bibinfo{person}{Fuli Feng}, \bibinfo{person}{Yitong Wang}, \bibinfo{person}{Xiangnan He}, \bibinfo{person}{Yong Liao}, {and} \bibinfo{person}{Yongdong Zhang}.} \bibinfo{year}{2022}\natexlab{}.
\newblock \showarticletitle{Addressing Unmeasured Confounder for Recommendation with Sensitivity Analysis}. In \bibinfo{booktitle}{\emph{Proceedings of the 28th ACM SIGKDD Conference on Knowledge Discovery and Data Mining}}. ACM, \bibinfo{pages}{305–315}.
\newblock
\urldef\tempurl%
\url{https://doi.org/10.1145/3534678.3539240}
\showDOI{\tempurl}


\bibitem[Gao et~al\mbox{.}(2022)]%
        {gao2022kuairand}
\bibfield{author}{\bibinfo{person}{Chongming Gao}, \bibinfo{person}{Shijun Li}, \bibinfo{person}{Yuan Zhang}, \bibinfo{person}{Jiawei Chen}, \bibinfo{person}{Biao Li}, \bibinfo{person}{Wenqiang Lei}, \bibinfo{person}{Peng Jiang}, {and} \bibinfo{person}{Xiangnan He}.} \bibinfo{year}{2022}\natexlab{}.
\newblock \showarticletitle{Kuairand: an unbiased sequential recommendation dataset with randomly exposed videos}. In \bibinfo{booktitle}{\emph{Proceedings of the 31st ACM International Conference on Information \& Knowledge Management}}. \bibinfo{pages}{3953--3957}.
\newblock


\bibitem[Guo et~al\mbox{.}(2017)]%
        {deepfm}
\bibfield{author}{\bibinfo{person}{Huifeng Guo}, \bibinfo{person}{Ruiming Tang}, \bibinfo{person}{Yunming Ye}, \bibinfo{person}{Zhenguo Li}, {and} \bibinfo{person}{Xiuqiang He}.} \bibinfo{year}{2017}\natexlab{}.
\newblock \showarticletitle{DeepFM: a factorization-machine based neural network for CTR prediction}.
\newblock \bibinfo{journal}{\emph{arXiv preprint arXiv:1703.04247}} (\bibinfo{year}{2017}).
\newblock


\bibitem[Hartford et~al\mbox{.}(2017)]%
        {iv-11}
\bibfield{author}{\bibinfo{person}{Jason Hartford}, \bibinfo{person}{Greg Lewis}, \bibinfo{person}{Kevin Leyton-Brown}, {and} \bibinfo{person}{Matt Taddy}.} \bibinfo{year}{2017}\natexlab{}.
\newblock \showarticletitle{Deep IV: A flexible approach for counterfactual prediction}. In \bibinfo{booktitle}{\emph{International Conference on Machine Learning}}. PMLR, \bibinfo{pages}{1414--1423}.
\newblock


\bibitem[He et~al\mbox{.}(2020)]%
        {lightgcn}
\bibfield{author}{\bibinfo{person}{Xiangnan He}, \bibinfo{person}{Kuan Deng}, \bibinfo{person}{Xiang Wang}, \bibinfo{person}{Yan Li}, \bibinfo{person}{Yongdong Zhang}, {and} \bibinfo{person}{Meng Wang}.} \bibinfo{year}{2020}\natexlab{}.
\newblock \showarticletitle{Lightgcn: Simplifying and powering graph convolution network for recommendation}. In \bibinfo{booktitle}{\emph{Proceedings of the 43rd International ACM SIGIR conference on research and development in Information Retrieval}}. \bibinfo{pages}{639--648}.
\newblock


\bibitem[Imbens and Rubin(2015)]%
        {imbens2015causal}
\bibfield{author}{\bibinfo{person}{Guido~W Imbens} {and} \bibinfo{person}{Donald~B Rubin}.} \bibinfo{year}{2015}\natexlab{}.
\newblock \bibinfo{booktitle}{\emph{Causal inference in statistics, social, and biomedical sciences}}.
\newblock \bibinfo{publisher}{Cambridge university press}.
\newblock


\bibitem[Khemakhem et~al\mbox{.}(2020)]%
        {9}
\bibfield{author}{\bibinfo{person}{Ilyes Khemakhem}, \bibinfo{person}{Diederik Kingma}, \bibinfo{person}{Ricardo Monti}, {and} \bibinfo{person}{Aapo Hyvarinen}.} \bibinfo{year}{2020}\natexlab{}.
\newblock \showarticletitle{Variational autoencoders and nonlinear ica: A unifying framework}. In \bibinfo{booktitle}{\emph{International Conference on Artificial Intelligence and Statistics}}. PMLR, \bibinfo{pages}{2207--2217}.
\newblock


\bibitem[Kingma and Ba(2014)]%
        {kingma2014adam}
\bibfield{author}{\bibinfo{person}{Diederik~P Kingma} {and} \bibinfo{person}{Jimmy Ba}.} \bibinfo{year}{2014}\natexlab{}.
\newblock \showarticletitle{Adam: A method for stochastic optimization}.
\newblock \bibinfo{journal}{\emph{arXiv preprint arXiv:1412.6980}} (\bibinfo{year}{2014}).
\newblock


\bibitem[Kmenta(2010)]%
        {iv-17}
\bibfield{author}{\bibinfo{person}{Jan Kmenta}.} \bibinfo{year}{2010}\natexlab{}.
\newblock \bibinfo{title}{Mostly harmless econometrics: An empiricist's companion}.
\newblock
\newblock


\bibitem[Koren et~al\mbox{.}(2009)]%
        {12}
\bibfield{author}{\bibinfo{person}{Yehuda Koren}, \bibinfo{person}{Robert Bell}, {and} \bibinfo{person}{Chris Volinsky}.} \bibinfo{year}{2009}\natexlab{}.
\newblock \showarticletitle{Matrix factorization techniques for recommender systems}.
\newblock \bibinfo{journal}{\emph{Computer}} \bibinfo{volume}{42}, \bibinfo{number}{8} (\bibinfo{year}{2009}), \bibinfo{pages}{30--37}.
\newblock


\bibitem[Kruskal(1989)]%
        {10.5555/120565.120567}
\bibfield{author}{\bibinfo{person}{J.~B. Kruskal}.} \bibinfo{year}{1989}\natexlab{}.
\newblock \bibinfo{booktitle}{\emph{Rank, Decomposition, and Uniqueness for 3-Way and n-Way Arrays}}.
\newblock \bibinfo{publisher}{North-Holland Publishing Co.}, \bibinfo{address}{NLD}, \bibinfo{pages}{7–18}.
\newblock
\showISBNx{0444874100}


\bibitem[Liu et~al\mbox{.}(2022)]%
        {15}
\bibfield{author}{\bibinfo{person}{Haochen Liu}, \bibinfo{person}{Da Tang}, \bibinfo{person}{Ji Yang}, \bibinfo{person}{Xiangyu Zhao}, \bibinfo{person}{Hui Liu}, \bibinfo{person}{Jiliang Tang}, {and} \bibinfo{person}{Youlong Cheng}.} \bibinfo{year}{2022}\natexlab{}.
\newblock \showarticletitle{Rating Distribution Calibration for Selection Bias Mitigation in Recommendations}. In \bibinfo{booktitle}{\emph{Proceedings of the ACM Web Conference 2022}}. Association for Computing Machinery, \bibinfo{pages}{2048--2057}.
\newblock
\urldef\tempurl%
\url{https://doi.org/10.1145/3485447.3512078}
\showDOI{\tempurl}


\bibitem[Marlin et~al\mbox{.}(2012)]%
        {select1}
\bibfield{author}{\bibinfo{person}{Benjamin Marlin}, \bibinfo{person}{Richard~S Zemel}, \bibinfo{person}{Sam Roweis}, {and} \bibinfo{person}{Malcolm Slaney}.} \bibinfo{year}{2012}\natexlab{}.
\newblock \showarticletitle{Collaborative filtering and the missing at random assumption}.
\newblock \bibinfo{journal}{\emph{arXiv preprint arXiv:1206.5267}} (\bibinfo{year}{2012}).
\newblock


\bibitem[Marlin and Zemel(2009)]%
        {yahoo}
\bibfield{author}{\bibinfo{person}{Benjamin~M Marlin} {and} \bibinfo{person}{Richard~S Zemel}.} \bibinfo{year}{2009}\natexlab{}.
\newblock \showarticletitle{Collaborative prediction and ranking with non-random missing data}. In \bibinfo{booktitle}{\emph{Proceedings of the third ACM conference on Recommender systems}}. \bibinfo{pages}{5--12}.
\newblock


\bibitem[Miao et~al\mbox{.}(2022)]%
        {17}
\bibfield{author}{\bibinfo{person}{Wang Miao}, \bibinfo{person}{Wenjie Hu}, \bibinfo{person}{Elizabeth~L Ogburn}, {and} \bibinfo{person}{Xiao-Hua Zhou}.} \bibinfo{year}{2022}\natexlab{}.
\newblock \showarticletitle{Identifying effects of multiple treatments in the presence of unmeasured confounding}.
\newblock \bibinfo{journal}{\emph{J. Amer. Statist. Assoc.}} (\bibinfo{year}{2022}), \bibinfo{pages}{1--15}.
\newblock


\bibitem[Pearl(2009)]%
        {20}
\bibfield{author}{\bibinfo{person}{Judea Pearl}.} \bibinfo{year}{2009}\natexlab{}.
\newblock \bibinfo{booktitle}{\emph{Causality: Models, Reasoning and Inference} (\bibinfo{edition}{2nd} ed.)}.
\newblock \bibinfo{publisher}{Cambridge University Press}.
\newblock


\bibitem[Robins(1986a)]%
        {ROBINS19861393}
\bibfield{author}{\bibinfo{person}{James Robins}.} \bibinfo{year}{1986}\natexlab{a}.
\newblock \showarticletitle{A new approach to causal inference in mortality studies with a sustained exposure period—application to control of the healthy worker survivor effect}.
\newblock \bibinfo{journal}{\emph{Mathematical Modelling}} \bibinfo{volume}{7}, \bibinfo{number}{9} (\bibinfo{year}{1986}), \bibinfo{pages}{1393--1512}.
\newblock
\showISSN{0270-0255}
\urldef\tempurl%
\url{https://doi.org/10.1016/0270-0255(86)90088-6}
\showDOI{\tempurl}


\bibitem[Robins(1986b)]%
        {21}
\bibfield{author}{\bibinfo{person}{James Robins}.} \bibinfo{year}{1986}\natexlab{b}.
\newblock \showarticletitle{A new approach to causal inference in mortality studies with a sustained exposure period—application to control of the healthy worker survivor effect}.
\newblock \bibinfo{journal}{\emph{Mathematical Modelling}} \bibinfo{volume}{7}, \bibinfo{number}{9} (\bibinfo{year}{1986}), \bibinfo{pages}{1393--1512}.
\newblock
\urldef\tempurl%
\url{https://doi.org/10.1016/0270-0255(86)90088-6}
\showDOI{\tempurl}


\bibitem[Schnabel et~al\mbox{.}(2016)]%
        {23}
\bibfield{author}{\bibinfo{person}{Tobias Schnabel}, \bibinfo{person}{Adith Swaminathan}, \bibinfo{person}{Ashudeep Singh}, \bibinfo{person}{Navin Chandak}, {and} \bibinfo{person}{Thorsten Joachims}.} \bibinfo{year}{2016}\natexlab{}.
\newblock \showarticletitle{Recommendations as treatments: Debiasing learning and evaluation}. In \bibinfo{booktitle}{\emph{International Conference on Machine Learning}}. PMLR, \bibinfo{pages}{1670--1679}.
\newblock


\bibitem[Si et~al\mbox{.}(2022)]%
        {IV-serach-data}
\bibfield{author}{\bibinfo{person}{Zihua Si}, \bibinfo{person}{Xueran Han}, \bibinfo{person}{Xiao Zhang}, \bibinfo{person}{Jun Xu}, \bibinfo{person}{Yue Yin}, \bibinfo{person}{Yang Song}, {and} \bibinfo{person}{Ji-Rong Wen}.} \bibinfo{year}{2022}\natexlab{}.
\newblock \showarticletitle{A Model-Agnostic Causal Learning Framework for Recommendation using Search Data}. In \bibinfo{booktitle}{\emph{Proceedings of the ACM Web Conference 2022}}. \bibinfo{pages}{224--233}.
\newblock


\bibitem[Si et~al\mbox{.}(2023)]%
        {si2023enhancing}
\bibfield{author}{\bibinfo{person}{Zihua Si}, \bibinfo{person}{Zhongxiang Sun}, \bibinfo{person}{Xiao Zhang}, \bibinfo{person}{Jun Xu}, \bibinfo{person}{Yang Song}, \bibinfo{person}{Xiaoxue Zang}, {and} \bibinfo{person}{Ji-Rong Wen}.} \bibinfo{year}{2023}\natexlab{}.
\newblock \showarticletitle{Enhancing recommendation with search data in a causal learning manner}.
\newblock \bibinfo{journal}{\emph{ACM Transactions on Information Systems}} \bibinfo{volume}{41}, \bibinfo{number}{4} (\bibinfo{year}{2023}), \bibinfo{pages}{1--31}.
\newblock


\bibitem[Venkatraman et~al\mbox{.}(2016)]%
        {iv-29}
\bibfield{author}{\bibinfo{person}{Arun Venkatraman}, \bibinfo{person}{Wen Sun}, \bibinfo{person}{Martial Hebert}, \bibinfo{person}{J Bagnell}, {and} \bibinfo{person}{Byron Boots}.} \bibinfo{year}{2016}\natexlab{}.
\newblock \showarticletitle{Online instrumental variable regression with applications to online linear system identification}. In \bibinfo{booktitle}{\emph{Proceedings of the AAAI Conference on Artificial Intelligence}}, Vol.~\bibinfo{volume}{30}.
\newblock


\bibitem[Wang et~al\mbox{.}(2017)]%
        {dcn}
\bibfield{author}{\bibinfo{person}{Ruoxi Wang}, \bibinfo{person}{Bin Fu}, \bibinfo{person}{Gang Fu}, {and} \bibinfo{person}{Mingliang Wang}.} \bibinfo{year}{2017}\natexlab{}.
\newblock \showarticletitle{Deep \& cross network for ad click predictions}.
\newblock In \bibinfo{booktitle}{\emph{Proceedings of the ADKDD'17}}. \bibinfo{pages}{1--7}.
\newblock


\bibitem[Wang et~al\mbox{.}(2020)]%
        {29}
\bibfield{author}{\bibinfo{person}{Yixin Wang}, \bibinfo{person}{Dawen Liang}, \bibinfo{person}{Laurent Charlin}, {and} \bibinfo{person}{David~M Blei}.} \bibinfo{year}{2020}\natexlab{}.
\newblock \showarticletitle{Causal inference for recommender systems}. In \bibinfo{booktitle}{\emph{Fourteenth ACM Conference on Recommender Systems}}. \bibinfo{pages}{426--431}.
\newblock


\bibitem[Wang et~al\mbox{.}(2022)]%
        {30}
\bibfield{author}{\bibinfo{person}{Zimu Wang}, \bibinfo{person}{Yue He}, \bibinfo{person}{Jiashuo Liu}, \bibinfo{person}{Wenchao Zou}, \bibinfo{person}{Philip~S Yu}, {and} \bibinfo{person}{Peng Cui}.} \bibinfo{year}{2022}\natexlab{}.
\newblock \showarticletitle{Invariant Preference Learning for General Debiasing in Recommendation}. In \bibinfo{booktitle}{\emph{Proceedings of the 28th ACM SIGKDD Conference on Knowledge Discovery and Data Mining}}. \bibinfo{pages}{1969--1978}.
\newblock


\bibitem[Wu et~al\mbox{.}(2022)]%
        {32}
\bibfield{author}{\bibinfo{person}{Le Wu}, \bibinfo{person}{Xiangnan He}, \bibinfo{person}{Xiang Wang}, \bibinfo{person}{Kun Zhang}, {and} \bibinfo{person}{Meng Wang}.} \bibinfo{year}{2022}\natexlab{}.
\newblock \showarticletitle{A survey on accuracy-oriented neural recommendation: From collaborative filtering to information-rich recommendation}.
\newblock \bibinfo{journal}{\emph{IEEE Transactions on Knowledge and Data Engineering}} (\bibinfo{year}{2022}).
\newblock


\bibitem[Xu et~al\mbox{.}(2020)]%
        {iv-36}
\bibfield{author}{\bibinfo{person}{Liyuan Xu}, \bibinfo{person}{Yutian Chen}, \bibinfo{person}{Siddarth Srinivasan}, \bibinfo{person}{Nando de Freitas}, \bibinfo{person}{Arnaud Doucet}, {and} \bibinfo{person}{Arthur Gretton}.} \bibinfo{year}{2020}\natexlab{}.
\newblock \showarticletitle{Learning deep features in instrumental variable regression}.
\newblock \bibinfo{journal}{\emph{arXiv preprint arXiv:2010.07154}} (\bibinfo{year}{2020}).
\newblock


\bibitem[Zhan et~al\mbox{.}(2022)]%
        {34}
\bibfield{author}{\bibinfo{person}{Ruohan Zhan}, \bibinfo{person}{Changhua Pei}, \bibinfo{person}{Qiang Su}, \bibinfo{person}{Jianfeng Wen}, \bibinfo{person}{Xueliang Wang}, \bibinfo{person}{Guanyu Mu}, \bibinfo{person}{Dong Zheng}, \bibinfo{person}{Peng Jiang}, {and} \bibinfo{person}{Kun Gai}.} \bibinfo{year}{2022}\natexlab{}.
\newblock \showarticletitle{Deconfounding Duration Bias in Watch-time Prediction for Video Recommendation}. In \bibinfo{booktitle}{\emph{Proceedings of the 28th ACM SIGKDD Conference on Knowledge Discovery and Data Mining}}. \bibinfo{pages}{4472--4481}.
\newblock


\bibitem[Zhang et~al\mbox{.}(2023)]%
        {idcf}
\bibfield{author}{\bibinfo{person}{Qing Zhang}, \bibinfo{person}{Xiaoying Zhang}, \bibinfo{person}{Yang Liu}, \bibinfo{person}{Hongning Wang}, \bibinfo{person}{Min Gao}, \bibinfo{person}{Jiheng Zhang}, {and} \bibinfo{person}{Ruocheng Guo}.} \bibinfo{year}{2023}\natexlab{}.
\newblock \showarticletitle{Debiasing Recommendation by Learning Identifiable Latent Confounders}.
\newblock \bibinfo{journal}{\emph{arXiv preprint arXiv:2302.05052}} (\bibinfo{year}{2023}).
\newblock


\bibitem[Zhang et~al\mbox{.}(2019)]%
        {r2}
\bibfield{author}{\bibinfo{person}{Shuai Zhang}, \bibinfo{person}{Lina Yao}, \bibinfo{person}{Aixin Sun}, {and} \bibinfo{person}{Yi Tay}.} \bibinfo{year}{2019}\natexlab{}.
\newblock \showarticletitle{Deep learning based recommender system: A survey and new perspectives}.
\newblock \bibinfo{journal}{\emph{ACM computing surveys (CSUR)}} \bibinfo{volume}{52}, \bibinfo{number}{1} (\bibinfo{year}{2019}), \bibinfo{pages}{1--38}.
\newblock


\bibitem[Zhang et~al\mbox{.}(2021)]%
        {35}
\bibfield{author}{\bibinfo{person}{Yang Zhang}, \bibinfo{person}{Fuli Feng}, \bibinfo{person}{Xiangnan He}, \bibinfo{person}{Tianxin Wei}, \bibinfo{person}{Chonggang Song}, \bibinfo{person}{Guohui Ling}, {and} \bibinfo{person}{Yongdong Zhang}.} \bibinfo{year}{2021}\natexlab{}.
\newblock \showarticletitle{Causal Intervention for Leveraging Popularity Bias in Recommendation}. In \bibinfo{booktitle}{\emph{Proceedings of the 44th International ACM SIGIR Conference on Research and Development in Information Retrieval}}. \bibinfo{pages}{11--20}.
\newblock


\bibitem[Zheng et~al\mbox{.}(2021)]%
        {conformity1}
\bibfield{author}{\bibinfo{person}{Yu Zheng}, \bibinfo{person}{Chen Gao}, \bibinfo{person}{Xiang Li}, \bibinfo{person}{Xiangnan He}, \bibinfo{person}{Yong Li}, {and} \bibinfo{person}{Depeng Jin}.} \bibinfo{year}{2021}\natexlab{}.
\newblock \showarticletitle{Disentangling user interest and conformity for recommendation with causal embedding}. In \bibinfo{booktitle}{\emph{Proceedings of the Web Conference 2021}}. \bibinfo{pages}{2980--2991}.
\newblock


\bibitem[Zhou et~al\mbox{.}(2018a)]%
        {din}
\bibfield{author}{\bibinfo{person}{Guorui Zhou}, \bibinfo{person}{Chengru Song}, \bibinfo{person}{Xiaoqiang Zhu}, \bibinfo{person}{Ying Fan}, \bibinfo{person}{Han Zhu}, \bibinfo{person}{Xiao Ma}, \bibinfo{person}{Yanghui Yan}, \bibinfo{person}{Junqi Jin}, \bibinfo{person}{Han Li}, {and} \bibinfo{person}{Kun Gai}.} \bibinfo{year}{2018}\natexlab{a}.
\newblock \showarticletitle{Deep interest network for click-through rate prediction}. In \bibinfo{booktitle}{\emph{Proceedings of the 24th ACM SIGKDD international conference on knowledge discovery \& data mining}}. \bibinfo{pages}{1059--1068}.
\newblock


\bibitem[Zhou et~al\mbox{.}(2018b)]%
        {36}
\bibfield{author}{\bibinfo{person}{Guorui Zhou}, \bibinfo{person}{Xiaoqiang Zhu}, \bibinfo{person}{Chenru Song}, \bibinfo{person}{Ying Fan}, \bibinfo{person}{Han Zhu}, \bibinfo{person}{Xiao Ma}, \bibinfo{person}{Yanghui Yan}, \bibinfo{person}{Junqi Jin}, \bibinfo{person}{Han Li}, {and} \bibinfo{person}{Kun Gai}.} \bibinfo{year}{2018}\natexlab{b}.
\newblock \showarticletitle{Deep interest network for click-through rate prediction}. In \bibinfo{booktitle}{\emph{Proceedings of the 24th ACM SIGKDD International Conference on Knowledge Discovery \& Data Mining}}. \bibinfo{pages}{1059--1068}.
\newblock


\bibitem[Zhu et~al\mbox{.}(2022b)]%
        {37}
\bibfield{author}{\bibinfo{person}{Xinyuan Zhu}, \bibinfo{person}{Yang Zhang}, \bibinfo{person}{Fuli Feng}, \bibinfo{person}{Xun Yang}, \bibinfo{person}{Dingxian Wang}, {and} \bibinfo{person}{Xiangnan He}.} \bibinfo{year}{2022}\natexlab{b}.
\newblock \showarticletitle{Mitigating hidden confounding effects for causal recommendation}.
\newblock \bibinfo{journal}{\emph{arXiv preprint arXiv:2205.07499(2022)}} (\bibinfo{year}{2022}).
\newblock


\bibitem[Zhu et~al\mbox{.}(2022a)]%
        {38}
\bibfield{author}{\bibinfo{person}{Yaochen Zhu}, \bibinfo{person}{Jing Yi}, \bibinfo{person}{Jiayi Xie}, {and} \bibinfo{person}{Zhenzhong Chen}.} \bibinfo{year}{2022}\natexlab{a}.
\newblock \showarticletitle{Deep causal reasoning for recommendations}.
\newblock \bibinfo{journal}{\emph{arXiv preprint arXiv:2201.02088}} (\bibinfo{year}{2022}).
\newblock


\end{thebibliography}

\appendix

\section{Poof of Theorem 1}
\label{proof}
 \begin{theorem}
\label{Theorem 1} 
Given a causal DAG $\mathcal{G} = (\mathbf{Z} \cup \mathbf{W} \cup \mathbf{C} \cup \mathbf{B} \cup \{\mathbf{T}, \mathbf{A}, \mathbf{R}\}, \mathbf{E})$, where $\mathbf{T}$ and $\mathbf{R}$ are the treatment and outcome, respectively, and $\mathbf{E}$ is the set of edges between the variables. Let $\mathbf{B}$ denote the latent confounders between $\mathbf{T}$ and $\mathbf{R}$. Suppose there exists a directed edge $\mathbf{T} \rightarrow \mathbf{R}$ in $\mathbf{E}$, and  $\mathbf{Z}$ represents the embeddings of user features. Additionally, assume that the underlying data generating process is faithful to the proposed causal DAG $\mathcal{G}$ as shown in Figure \ref{fig:2}. Therefore, $\mathbf{Z}$ serves as the valid IV for estimating the causal effect of $\mathbf{T}$ on $\mathbf{R}$.
\end{theorem}

\begin{proof}
First, we assume that the data-generating process is faithful to the causal DAG $\mathcal{G}$ in Figure \ref{fig:2}, meaning that the conditional dependencies and independencies between the variables in the data can be accurately read from the DAG. Thus, we prove that the embeddings of user features, $\mathbf{Z}$, satisfy the three conditions of Definition \hyperref[Definition 1]{1} based on the causal DAG shown in Figure \ref{fig:2}. (1) In the causal DAG $\mathcal{G}$, there exists a direct edge $\mathbf{Z} \rightarrow \mathbf{T}$, which implies that $\mathbf{Z} \nindep_d \mathbf{T}$. Therefore, the first condition of Definition \hyperref[Definition 1]{1} holds.
(2) In the causal DAG $\mathcal{G}$, there are no direct paths from $\mathbf{Z}$ to $\mathbf{R}$ outside of $\mathbf{T}$, we have $\mathbf{Z} \indep_d \mathbf{R} \mid \{\mathbf{T}, \mathbf{B}\}$. This satisfies the second condition of Definition \hyperref[Definition 1]{1}. (3) Lastly, we establish that there is no confounding bias between $\mathbf{Z}$ and the outcome $\mathbf{R}$. In the causal DAG $\mathcal{G}$, there are no any back-door paths from  $\mathbf{Z}$ to $\mathbf{R}$, i.e.,  $(\mathbf{Z} \indep_{d} \mathbf{R})_{\mathcal{G}_{\underline{\mathbf{T}}}}$.  Hence, the third condition of Definition \hyperref[Definition 1]{1} is satisfied.

Therefore, we conclude that $\mathbf{Z}$ satisfies all three conditions required for valid instrumental variables.
\end{proof}

\section{Algorithm}
\label{appendix:A}

\begin{algorithm}
    \SetKwInOut{Input}{Input}\SetKwInOut{Output}{Output}
    \caption{Instrumental Variables-based Identifiable Disentangled
Debiased Learning (IViDR).}
    \label{alg}
    \Input{$\{\mathbf{X}, i\}, \forall i \in \mathbf{I}$, $\{\mathbf{A}, \mathbf{W}\},\forall u \in \mathbf{U}$, $\{r_{ui}\}, \forall (u,i) \in \mathcal{D}$ }
    
    \tcp{Treatment reconstruction using IVs}

    Reconstruct treatment $\mathbf{T}$ by instrumental variables $\mathbf{Z}$ to get the reconstructed treatment $\mathbf{T}^{\mathrm{re}}$.

    Fuse the input $\mathbf{X}$ of iVAE with the reconstructed treatment $\mathbf{T}^{\mathrm{re}}$ to obtain the debiased interactional data $\mathbf{X}^{\mathrm{re}}$.

    $\mathbf{X}^{\mathrm{re}} \leftarrow \mathbf{X} + \mathbf{T}_j^{\mathrm{re}}$
        
    \tcp{{Learning latent confounders}}
    
    Calculate the latent confounders distribution $q_\phi\left({\mathbf{C}}_1 \mid \mathbf{A}, \mathbf{W}\right)$ for each user $u$ by maximizing Eq.\cref{eq:13}\;

    Calculate the latent confounders distribution $q_\phi\left({\mathbf{C}}_2 \mid \mathbf{A}, \mathbf{W}\right)$ for each user $u$ by maximizing Eq.\cref{eq:13}\;

    Fuse $q_\phi\left({\mathbf{C}}_1 \mid \mathbf{A}, \mathbf{W}\right)$ and $q_\phi\left({\mathbf{C}}_2 \mid \mathbf{A}, \mathbf{W}\right)$ to get latent confounders $q_\phi\left(\mathbf{C} \mid \mathbf{A}, \mathbf{W}\right)$ ($\rho, \tau$ are tuning parameters).

    $\mathbf{C} \leftarrow \rho * {{\mathbf{C}}_1}  +  \tau *{{\mathbf{C}}_2}$
    
   \tcp{Training recommendation model}
   
    Initialize a recommendation model $f(u, i,\mathbf{C}; \eta)$ with parameters $\eta$\;
    \While{Stop condition is not reached}
    {
    Fetch $(u,i)$ from $\mathcal{D}$\;
    Minimize the loss Eq.\cref{eq:18} to optimize $\eta$;
    }
   
\end{algorithm}

The algorithm consists of three main steps.
The first step is treatment reconstruction using IVs.
We use $\mathbf{Z}$ as the IV to reconstruct the treatment, obtaining $\mathbf{T}^{\mathrm{re}}$.
We then combine the input $\mathbf{X}$ of iVAE with $\mathbf{T}^{\mathrm{re}}$ to obtain the debiased interactional data $\mathbf{X}^{\mathrm{re}}$.
The second step is learning latent confounders.
We use the debiased interactional data $\mathbf{X}^{\mathrm{re}}$ to calculate $q_\phi\left({\mathbf{C}}_1 \mid \mathbf{A}, \mathbf{W}\right)$ for each user $u$ by maximizing Eq.\cref{eq:13}\/.
We use the input $\mathbf{X}$ of iVAE to calculate $q_\phi\left({\mathbf{C}}_2 \mid \mathbf{A}, \mathbf{W}\right)$ for each user $u$ by maximizing Eq.\cref{eq:13}\/.
We then combine \( q_\phi\left({\mathbf{C}}_1 \mid \mathbf{A}, \mathbf{W}\right) \vspace{1pt} \) and \( q_\phi\left({\mathbf{C}}_2 \mid \mathbf{A}, \mathbf{W}\right) \) to get latent confounder \( q_\phi\left(\mathbf{C} \mid \mathbf{A}, \mathbf{W}\right) \).
The third step is training the recommendation model.
We adjust for the learned latent representation to mitigate confounding bias.

\section{Detailed proofs of the three steps of Theorem 2.}
\label{appendix:B}

\begin{proof}[Proof of Theorem~\ref{Theorem 2}]

\textbf{Step I:}
We introduce here the volume of the matrix, denoted $\operatorname{vol} M$, which is the product of the singular values of $M$. 
When $M$ has full column rank, $\operatorname{vol} M=\sqrt{\operatorname{det} M^T M}$, and when $M$ is invertible, $\operatorname{vol} M=|\operatorname{det} M|$.
The matrix volume can substitute for Jacobi's absolute determinant in the variable transformation formula. 
This is especially useful when the Jacobian is a rectangular matrix $(n<d)$. 
Suppose we have two sets of parameters: $(\mathbf{f}, \mathbf{H}, \mathbf{\lambda})$ and $(\tilde{\mathbf{f}}, \tilde{\mathbf{H}}, \tilde{\mathbf{\lambda}})$ such that $p_{\mathbf{f}, \mathbf{H}, \mathbf{\lambda}}(\mathbf{X}^{\mathrm{re}} \mid \mathbf{W})=p_{\tilde{\mathbf{f}}, \tilde{\mathbf{H}}, \tilde{\mathbf{\lambda}}}(\mathbf{X}^{\mathrm{re}} \mid \mathbf{W})$ for all pairs $(\mathbf{X}^{\mathrm{re}}, \mathbf{W})$. 
Then:

\begin{equation}
\begin{split}
\label{zm4}
&\int_{c} p_{\mathbf{H}, \mathbf{\lambda}}(\mathbf{C} \mid \mathbf{W}) p_{\mathbf{f}}(\mathbf{X}^{\mathrm{re}} \mid \mathbf{C}) \mathrm{d} c
\\& =\int_{c} p_{\tilde{\mathbf{H}}, \tilde{\mathbf{\lambda}}}(\mathbf{C} \mid \mathbf{W}) p_{\tilde{\mathbf{f}}}(\mathbf{X}^{\mathrm{re}} \mid \mathbf{C}) \mathrm{d} c
\end{split}
\end{equation}

According to Eq. \cref{zm2}, we get:

\begin{equation}
\begin{split}
\label{zm5}
&\int_{c} p_{\mathbf{H}, \mathbf{\lambda}}(\mathbf{C} \mid \mathbf{W}) p_{\varepsilon}(\mathbf{X}^{\mathrm{re}}-\mathbf{f}(\mathbf{C})) \mathrm{d} c 
\\& = \int_{c} p_{\tilde{\mathbf{H}}, \tilde{\mathbf{\lambda}}}(\mathbf{C} \mid \mathbf{W}) p_{\varepsilon}(\mathbf{X}^{\mathrm{re}}-\tilde{\mathbf{f}}(\mathbf{C})) \mathrm{d} c
\end{split}
\end{equation}

In equation Eq. \cref{zm5}, we have made variable substitutions $\overline{\mathbf{X}}^{\mathrm{re}}=\mathbf{f}(\mathbf{C})$ on the left hand side, $\overline{\mathbf{X}}^{\mathrm{re}}=\tilde{\mathbf{f}}(\mathbf{C})$ on the right hand side, and $J$ denotes the Jacobian. We get:

\begin{equation}
\begin{split}
\label{zm6}
&\int_{\mathcal{X}} p_{\mathbf{H}, \mathbf{\lambda}}\left(\mathbf{f}^{-1}(\overline{\mathbf{X}}^{\mathrm{re}}) \mid \mathbf{W}\right) \operatorname{vol} J_{\mathbf{f}^{-1}}(\overline{\mathbf{X}}^{\mathrm{re}}) 
p_{\varepsilon}(\mathbf{X}^{\mathrm{re}}-\overline{\mathbf{X}}^{\mathrm{re}}) \mathrm{d} \overline{\mathbf{X}}^{\mathrm{re}}  \\&
=\int_{\mathcal{X}} p_{\tilde{\mathbf{H}}, \tilde{\mathbf{\lambda}}}\left(\tilde{\mathbf{f}}^{-1}(\overline{\mathbf{X}}^{\mathrm{re}}) \mid \mathbf{W}\right) \operatorname{vol} J_{\tilde{\mathbf{f}}^{-1}}(\overline{\mathbf{X}}^{\mathrm{re}}) 
p_{\varepsilon}(\mathbf{X}^{\mathrm{re}}-\overline{\mathbf{X}}^{\mathrm{re}}) \mathrm{d} \overline{\mathbf{X}}^{\mathrm{re}}
\end{split}
\end{equation}

In Eq.\cref{zm6}, we introduced:

\begin{equation}
\begin{split}
\label{zm7}
\tilde{p}_{\mathbf{H}, \mathbf{\lambda}, \mathbf{f}, \mathbf{W}}(\mathbf{X}^{\mathrm{re}})=p_{\mathbf{H}, \mathbf{\lambda}}\left(\mathbf{f}^{-1}(\mathbf{X}^{\mathrm{re}}) \mid \mathbf{W}\right) \operatorname{vol} J_{\mathbf{f}^{-1}}(\mathbf{X}^{\mathrm{re}}) 
\end{split}
\end{equation}
on the both sides. We get:

\begin{equation}
\begin{split}
\label{zm8}
&\int_{\mathbb{R}^d} \tilde{p}_{\mathbf{H}, \mathbf{\lambda}, \mathbf{f}, \mathbf{W}}(\overline{\mathbf{X}}^{\mathrm{re}}) p_{\varepsilon}(\mathbf{X}^{\mathrm{re}}-\overline{\mathbf{X}}^{\mathrm{re}}) \mathrm{d} \overline{\mathbf{X}}^{\mathrm{re}} 
\\&=\int_{\mathbb{R}^d} \tilde{p}_{\tilde{\mathbf{H}}, \tilde{\mathbf{\lambda}}, \tilde{\mathbf{f}}, \mathbf{W}}(\overline{\mathbf{X}}^{\mathrm{re}}) p_{\varepsilon}(\mathbf{X}^{\mathrm{re}}-\overline{\mathbf{X}}^{\mathrm{re}}) \mathrm{d} \overline{\mathbf{X}}^{\mathrm{re}}
\end{split}
\end{equation}

In Eq.\cref{zm8}, we use $*$ as the convolution operator. We get:

\begin{equation}
\begin{split}
\label{zm9}
\left(\tilde{p}_{\mathbf{H}, \mathbf{\lambda}, \mathbf{f}, \mathbf{W}} * p_{\varepsilon}\right)(\mathbf{X}^{\mathrm{re}})=\left(\tilde{p}_{\tilde{\mathbf{H}}, \tilde{\mathbf{\lambda}}, \tilde{\mathbf{f}}, \mathbf{W}} * p_{\varepsilon}\right)(\mathbf{X}^{\mathrm{re}})
\end{split}
\end{equation}

in Eq. \cref{zm9}, we used $F[$.$]$ to designate the Fourier transform, and where $\varphi_{\varepsilon}=F\left[p_{\varepsilon}\right]$.

\begin{equation}
\begin{split}
\label{zm10}
F\left[\tilde{p}_{\mathbf{H}, \mathbf{\lambda}, \mathbf{f}, \mathbf{W}}\right](\omega) \varphi_{\varepsilon}(\omega)=F\left[\tilde{p}_{\tilde{\mathbf{H}}, \tilde{\mathbf{\lambda}}, \tilde{\mathbf{f}}, \mathbf{W}}\right](\omega) \varphi_{\varepsilon}(\omega)
\end{split}
\end{equation}

In Eq. \cref{zm10}, we remove $\varphi_{\varepsilon}(\omega)$ from both sides, since it is non-zero almost everywhere (by assumption \hyperref[i]{(i)}).

\begin{equation}
\begin{split}
\label{zm11}
F\left[\tilde{p}_{\mathbf{H}, \mathbf{\lambda}, \mathbf{f}, \mathbf{W}}\right](\omega)=F\left[\tilde{p}_{\tilde{\mathbf{H}}, \tilde{\mathbf{\lambda}}, \tilde{\mathbf{f}}, \mathbf{W}}\right](\omega)
\end{split}
\end{equation}

According to the inverse transformation of the Fourier transform, we get:

\begin{equation}
\begin{split}
\label{zm12}
\tilde{p}_{\mathbf{H}, \mathbf{\lambda}, \mathbf{f}, \mathbf{W}}(\mathbf{X}^{\mathrm{re}})=\tilde{p}_{\tilde{\mathbf{H}}, \tilde{\mathbf{\lambda}}, \tilde{\mathbf{f}}, \mathbf{W}}(\mathbf{X}^{\mathrm{re}})
\end{split}
\end{equation}

\textbf{Step II:}
By taking the logarithm on both sides of equation Eq. \cref{zm12} and replacing $p_{H, \lambda}$ by its expression from Eq. \cref{zm3}, we get:

\begin{equation}
\begin{split}
\label{zm14}
& \log \operatorname{vol} J_{\mathbf{f}^{-1}}(\mathbf{X}^{\mathrm{re}}) +\sum_{i=1}^n
\\&  \left(\log Q_i\left(f_i^{-1}(\mathbf{X}^{\mathrm{re}})\right)-\log Z_i(\mathbf{W}) +\sum_{j=1}^k H_{i, j}\left(f_i^{-1}(\mathbf{X}^{\mathrm{re}})\right) \lambda_{i, j}(\mathbf{W})\right) \\
& = \log \operatorname{vol} J_{\hat{\mathbf{f}}^{-1}}(\mathbf{X}^{\mathrm{re}}) +\sum_{i=1}^n
\\& \left(\log \tilde{Q}_i\left(\tilde{f}_i^{-1}(\mathbf{X}^{\mathrm{re}})\right)-\log \tilde{Z}_i(\mathbf{W})  +\sum_{j=1}^k \tilde{H}_{i, j}\left(\tilde{f}_i^{-1}(\mathbf{X}^{\mathrm{re}})\right) \tilde{\lambda}_{i, j}(\mathbf{W})\right)
\end{split}
\end{equation}

Let $W_0, \ldots, W_{n k}$ be the points provided by assumption \hyperref[iv]{(iv)}, and define $\overline{\mathbf{\lambda}}(\mathbf{W})=\mathbf{\lambda}(\mathbf{W})-\mathbf{\lambda}\left(\mathbf{W}_0\right)$.
We substitute each $\mathbf{W}_l$ into Eq. \cref{zm14} to obtain $n k+1$ such equations. We subtract the first equation of $\mathbf{W}_0$ from the remaining $n k$ equations to obtain the following for $l=1, \ldots, n k$.

\begin{equation}
\begin{split}
\label{zm15}
&\left\langle\mathbf{H}\left(\mathbf{f}^{-1}(\mathbf{X}^{\mathrm{re}})\right), \overline{\mathbf{\lambda}}\left(\mathbf{W}_l\right)\right\rangle+\sum_i \log \frac{Z_i\left(\mathbf{W}_0\right)}{Z_i\left(\mathbf{W}_l\right)} 
\\& =\left\langle\tilde{\mathbf{H}}\left(\tilde{\mathbf{f}}^{-1}(\mathbf{X}^{\mathrm{re}})\right), \overline{\tilde{\lambda}}\left(\mathbf{W}_l\right)\right\rangle+\sum_i \log \frac{\tilde{Z}_i\left(\mathbf{W}_0\right)}{\tilde{Z}_i\left(\mathbf{W}_l\right)}
\end{split}
\end{equation}

Let $L$ denote the matrix defined in assumption \hyperref[iv]{(iv)}, and $\tilde{L}$ define $\tilde{\lambda}$ in a similar way ($\tilde{L}$ is not necessarily invertible).
Define $s_l=\sum_i \log \frac{\tilde{Z}_i\left(\mathbf{W}_0\right) Z_i\left(\mathbf{W}_l\right)}{Z_i\left(\mathbf{W}_0\right) \tilde{Z}_i\left(\mathbf{W}_l\right)}$, and let $\mathbf{s}$ be the vector of all $s_l$, where $l=1, \ldots, n k$.
Expressing Eq. \cref{zm15} as a matrix form for all points $\mathbf{W}_l$, we get:

\begin{equation}
\begin{split}
\label{zm16}
L^T \mathbf{H}\left(\mathbf{f}^{-1}(\mathbf{X}^{\mathrm{re}})\right)=\tilde{L}^T \tilde{\mathbf{H}}\left(\tilde{\mathbf{f}}^{-1}(\mathbf{X}^{\mathrm{re}})\right)+\mathbf{s}
\end{split}
\end{equation}

Then we multiply both sides of the above equations by $L^{-T}$:

\begin{equation}
\begin{split}
\label{zm17}
\mathbf{H}\left(\mathbf{f}^{-1}(\mathbf{X}^{\mathrm{re}})\right)=M \tilde{\mathbf{H}}\left(\tilde{\mathbf{f}}^{-1}(\mathbf{X}^{\mathrm{re}})\right)+\mathbf{q}
\end{split}
\end{equation}
where $M=L^{-T} \tilde{L}^{T}$ and $\mathbf{q}=L^{-T} \mathbf{s}$.


To prove that M is invertible, we first introduce Lemma \hyperref[Lemma 1:]{1}:

\begin{lemma}
\label{Lemma 1:}
Consider a strong exponential distribution of size $k \geq 2$ with sufficient statistics $\mathbf{H}(\mathbf{X}^{\mathrm{re}})=\left(H_1(\mathbf{X}^{\mathrm{re}}), \ldots, H_k(\mathbf{X}^{\mathrm{re}})\right)$.
Suppose further that $\mathbf{H}$ is differentiable almost everywhere. Then there exist $k$ distinct values ${{X}_1}^{\mathrm{re}}$ through ${{X}_k}^{\mathrm{re}}$ such that $\left(\mathbf{H}^{\prime}\left({{X}_1}^{\mathrm{re}}\right), \ldots, \mathbf{H}^{\prime}\left({{X}_k}^{\mathrm{re}}\right)\right) $ is linearly independent in $\mathbb{R}^k$.

\end{lemma}

\begin{proof}[Proof of Lemma~\ref{Lemma 1:}]
Suppose that for any chosen $k$ points, the family $\left(\mathbf{H}^{\prime}\left({{X}_1}^{\mathrm{re}}\right), \ldots, \mathbf{H}^{\prime}\left({{X}_k}^{\mathrm{re}}\right)\right)$ is never linearly independent.
This means that $\mathbf{H}^{\prime}(\mathbb{R})$ is contained in a subspace of $\mathbb{R}^k$ of dimension at most $k-1$.
Let $\mathbf{\theta}$ be a nonzero vector orthogonal to $\mathbf{H}^{\prime}(\mathbb{R})$.
Then for all $\mathbf{X}^{\mathrm{re}} \in \mathbb{R}$, we have $\left\langle\mathbf{H}^{\prime}(\mathbf{X}^{\mathrm{re}}), \mathbf{\theta}\right\rangle=0$.
We find $\langle \mathbf{H}(\mathbf{X}^{\mathrm{re}}), \mathbf{\theta} \rangle=$ constant by integration.
Since this holds for all $\mathbf{X}^{\mathrm{re}} \in \mathbb{R}$ and $\mathbf{\theta} \neq 0$, we conclude that the distribution is not strongly exponential, which contradicts Assumption \ref{assumption 1}.

\end{proof}

\textbf{Step III:}
Now, by definition of $\mathbf{H}$ and assumption \hyperref[iii]{(iii)}, its Jacobian matrix exists and is a matrix of $n k \times n$ with rank $n$. This means that the Jacobian matrix of $\tilde{\mathbf{H}} \circ \tilde{\mathbf{f}}^{-1}$ exists and has a rank of $n$, and so does $M$. We distinguish two cases: 
(1) If $k=1$, then this means that $M$ is invertible (because $M$ is $n \times n$).
(2) If $k>1$, we define 
$\mathbf{H}_i\left(\overline{{X}_i}^{\mathrm{re}}\right)=\left(H_{i, 1}\left(\overline{{X}_i}^{\mathrm{re}}\right), \ldots H_{i, k}\left(\overline{{X}_i}^{\mathrm{re }}\right)\right)$ and $\overline{\mathbf{X}}^{\mathrm{re}}=\mathbf{f}^{-1}(\mathbf{X}^{\mathrm{re}})$. 
According to Lemma \hyperref[Lemma 1:]{1}, for each $i \in[1, \ldots, n]$, there exist $k$ points $\overline{{X}_i}^{\mathrm{re1}}, \ldots, \overline{{X}_i}^{\mathrm{re}k}$ such that $\left(\mathbf{H}_i^{\prime}\left(\overline{{X}_i}^{\mathrm{re1}}\right), \ldots, \mathbf{H}_i^{\prime}\left(\overline{{X}_i}^{\mathrm{re}k}\right)\right)$ linearly independent.
Aggregate these points into $k$ vectors $\left(\overline{\mathbf{X}}^{\mathrm{re1}}, \ldots, \overline{\mathbf{X}}^{\mathrm{re}k}\right)$, and splice the computed $k$ Jacobian matrices $J_{\mathbf{H}}\left( \overline{\mathbf{X}}^{\mathrm{re}l}\right)$ is horizontally spliced into a matrix $Q=\left(J_{\mathbf{H}}\left(\overline{\mathbf{X}}^{\mathrm{re1}}\right), \ldots, J_{\mathbf{H}}\left(\overline{\mathbf{X}}^{\mathrm{re}k}\right)\right)$ (also define $\tilde{Q}$ as $\tilde{\mathbf{H}}(\tilde{\mathbf{f}}^{-1} \circ \mathbf{f}(\overline{\mathbf{X}}^{\mathrm{re}})$) of the Jacobi matrix splice).
Then the matrix $Q$ is invertible (by Lemma \hyperref[Lemma 1:]{1} and the combination of the fact that each component of $\tilde{H}$ is univariate).
By deriving Eq. \cref{zm17} for each $\mathbf{\mathbf{X}^{\mathrm{re}}}^l$, we get (in matrix form):

\begin{equation}
\begin{split}
Q=M \tilde{Q}
\end{split}
\end{equation}

The reversibility of $Q$ implies the reversibility of $M$ and $\tilde{Q}$.

Thus, Eq. \cref{zm17} and the invertibility of $M$ implies that $(\tilde{\mathbf{f}}, \tilde{\mathbf{H}}, \tilde{\mathbf{\lambda}}) \sim(\mathbf{f}, \mathbf{H}, \mathbf{\lambda})$, thus completing the proof.

\end{proof}

\section{IDENTIFICATION OF LATENT CONFOUNDER USING IVAE}
\label{appendix:C}

\subsection{The Detail of Obtaining the Approximate Posterior of Latent Confounders.}
\label{appendix:C2}
  
Following the standard iVAE \cite{9}, we use $q_\phi(\mathbf{C} \mid \mathbf{A}, \mathbf{W})$ as the approximate posterior to learn $p_\theta(\mathbf{C} \mid \mathbf{A}, \mathbf{W})$.

\begin{small} 
\begin{equation}
\label{eq:13}
\begin{split}
& E\left[\log p_\theta\left(\mathbf{A} \mid \mathbf{W}\right)\right] \geq \mathcal{L}(\theta, \phi) \\
= & E[\underbrace{E_{q_\phi\left(\mathbf{C} \mid \mathbf{A}, \mathbf{W}\right)}\left[\log p_\theta\left(\mathbf{C} \mid \mathbf{W}\right)-\log q_\phi\left(\mathbf{C} \mid \mathbf{A}, \mathbf{W}\right)\right]}_I 
\\& +\underbrace{E_{q_\phi\left(\mathbf{C} \mid \mathbf{A}, \mathbf{W}\right)}\left[\log p_\theta\left(\mathbf{A} \mid \mathbf{C}\right)\right]}_{I I}]
\end{split}
\end{equation}
\end{small}

We can further decompose $\log p_\theta\left(\mathbf{A}, \mathbf{C} \mid \mathbf{W}\right)$ into the following form.

\begin{equation}
\begin{split}
&\log p_\theta\left(\mathbf{A}, \mathbf{C} \mid \mathbf{W}\right)  
\\& =\log p_\theta\left(\mathbf{A} \mid \mathbf{C}, \mathbf{W}\right)+\log p_\theta\left(\mathbf{C} \mid \mathbf{W}\right) 
\\& =\log p_\theta\left(\mathbf{A} \mid \mathbf{C}\right)+\log p_\theta\left(\mathbf{C} \mid \mathbf{W}\right) 
\end{split}
\end{equation}

We employ the Gaussian distribution as the prior distribution for the representations, defined as:

\begin{equation}
\begin{split}
& p_\theta\left(\mathbf{C} \mid \mathbf{W}\right) :=N\left(\mu_w\left(\mathbf{W}\right), v_w\left(\mathbf{W}\right)\right), 
\\&
q_\phi\left(\mathbf{C} \mid \mathbf{A}, \mathbf{W}\right)
:= N\left(\mu_{a w}\left(\mathbf{A}, \mathbf{W}\right), v_{a w}\left(\mathbf{A}, \mathbf{W}\right)\right)
\end{split}
\end{equation}
where $N(\cdot, \cdot)$ the Gaussian distribution. 
Thus, the computation of the expectation \(I\) from equation Eq.~\cref{eq:13} can be simplified to the following equation:

\begin{equation}
\begin{aligned}
& E_{q_\phi\left(\mathbf{C} \mid \mathbf{A}, \mathbf{W}\right)}\left[\log p_\theta\left(\mathbf{C} \mid \mathbf{W}\right) - \log q_\phi\left(\mathbf{C} \mid \mathbf{A}, \mathbf{W}\right)\right] \\
&= -KL\left(N\left(\mu_{aw}\left(\mathbf{A}, \mathbf{W}\right), v_{aw}\left(\mathbf{A}, \mathbf{W}\right)\right), \right.  \left. N\left(\mu_w\left(\mathbf{W}\right), v_w\left(\mathbf{W}\right)\right)\right)
\end{aligned}
\end{equation}

For computing the  $II$ part in Eq.\cref{eq:13}, we have  the following equation:

\begin{equation}
\begin{split}
&\log p_\theta\left(\mathbf{A} \mid \mathbf{C}\right)
\\& =\sum_{i=1}^n {a}_{u i} \log \left(\mu_c\left(\mathbf{C}\right)_i\right)+  \left(1-{a}_{u i}\right) \log \left(1-\mu_c\left(\mathbf{C}\right)_i\right)
\end{split}
\end{equation}

Thus, by maximizing Eq.\cref{eq:13}, we can derive an approximate posterior $q_\phi({\mathbf{C}} \mid \mathbf{A}, \mathbf{W})$ for recovering latent confounder $\mathbf{C}$.

\section{EXPERIMENTAL DETAILS.}
\label{appendix:D}

\begin{figure*}

    \centering
    \begin{minipage}[a]{0.45\linewidth}
        \centering
        \includegraphics[width=\linewidth]{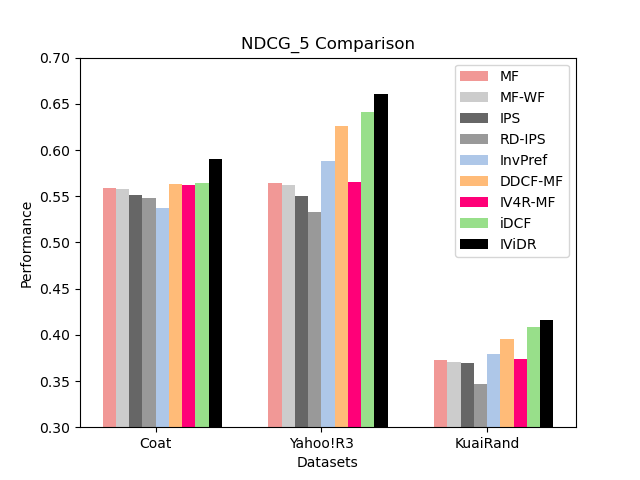}
         \captionsetup{labelformat=empty} 
        \vspace{0cm}  
    \end{minipage}
    \begin{minipage}[a]{0.45\linewidth}
        \centering
        \includegraphics[width=\linewidth]{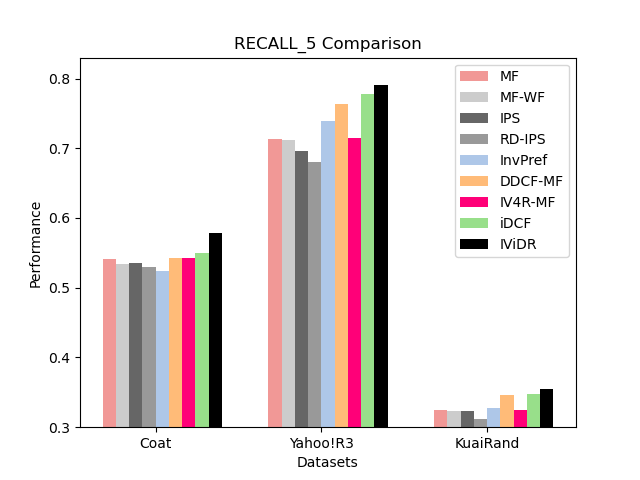}
         \captionsetup{labelformat=empty} 
    \end{minipage}
    
    \caption{The performance of all methods on the three real-world datasets.}
    \label{fig:4}
    \vspace{-10pt}
\end{figure*}

\subsection{Dataset description.}
\label{appendix:D1}

\noindent\textbf{Coat Dataset.} 
The Coat dataset \cite{23} simulates MNAR data from online coat purchases.
It provides user features and ratings (5-point scale, where 1 indicates the lowest rating).

\noindent\textbf{Yahoo!R3 Dataset.} 
The Yahoo!R3 dataset \cite{yahoo} contains user-item interactions from Yahoo!'s recommendation system.
It also provides user features and ratings (5-point scale, where 1 indicates the lowest rating).

\noindent\textbf{KuaiRand Dataset.} 
The KuaiRand dataset \cite{gao2022kuairand} includes 23,533 users and 6,712 videos.
It also provides user features and signals for whether the user clicked (IsClick=1 denoting a click).

\section{Data Generation Process}
\label{appendix:F}

The simulated dataset includes 10,000 users and 1000 items. 
For each user $u$, $\mathbf{C}$ is represented as a 2D vector.
This vector reflects the user’s socio-economic status and is drawn from a mixture of five independent multivariate Gaussian distributions.
For each item $i$, $\mathbf{V}$ is represented as a 2D vector reflecting the mixing of item factors, drawn from a mixture of five independent multivariate Gaussian distributions. 
The proxy variables $\mathbf{W}$ and $\mathbf{M}$ are one-dimensional categorical variables with probabilities determined by a Gaussian-like distribution.
The embeddings of user features $\mathbf{Z}$ are randomly distributed.
The conditional distribution of $\mathbf{C}$ follows:

\begin{equation}
\mathbf{C}^k \mid \mathbf{W} \sim N\left(\mu_k\left(\mathbf{W}\right), \sigma_k^2\left(\mathbf{W}\right)\right), k \in\{1,2\},
\end{equation}
where $\mathbf{C}{ }^k$ is the j-th element of $\mathbf{C}$.

The conditional distribution of $\mathbf{V}$ follows:

\begin{equation}
\mathbf{V}^k \mid \mathbf{M} \sim N\left(\mu_k\left(\mathbf{M}\right), \sigma_k^2\left(\mathbf{M}\right)\right), k \in\{1,2\},
\end{equation}
where $\mathbf{V}{ }^k$ is the j-th element of $\mathbf{V}$.

We can reconstruct $\mathbf{C}$ with $\mathbf{V}$ to get the reconstructed latent confounders $\mathbf{H}$.

The $(u, i)$ is generated by

\begin{equation}
\begin{aligned}
& a_{u i} \sim \operatorname{Bernoulli}\left(g_i\left(\mathbf{H}\right)\right), \\
& g_i(\mathbf{H})=\alpha \cdot \operatorname{sigmoid}\left(\text { LeakyRelu }\left(\mathbf{H} M e_{\mathbf{H} i}\right)+\gamma \epsilon\right)
\end{aligned}
\end{equation}
where $M$ is a 2 × 2 matrix, $e_{\mathbf{H} i}$ is a randomly generated item-wise 2D embedding vector, $\alpha$ is a hyperparameter, $\epsilon$ is random noise, and $\gamma$ is the corresponding weight of the noise.

The $r_{u i}=f_n\left(e_u^T e_i+\beta \mathbf{H}^T e_{H i}+\right.$ $\epsilon_{u i})$, where $f_n: \mathbb{R} \rightarrow\{1,2,3,4,5\}$ is a normalization function, $\epsilon_{u i}$ is an i.i.d. random noise, and $\beta$ is a hyperparameter.

\section{Ablation Study and Hyper-parameter Analysis}
\label{appendix:E}

\begin{figure*}

    \centering
    \begin{minipage}[a]{0.45\linewidth}
        \centering
        \includegraphics[width=\linewidth]{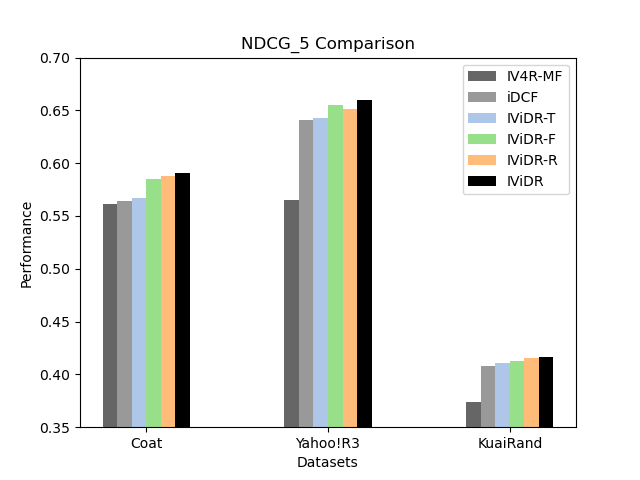}
         \captionsetup{labelformat=empty} 
        \vspace{0cm}  
    \end{minipage}
    \begin{minipage}[a]{0.45\linewidth}
        \centering
        \includegraphics[width=\linewidth]{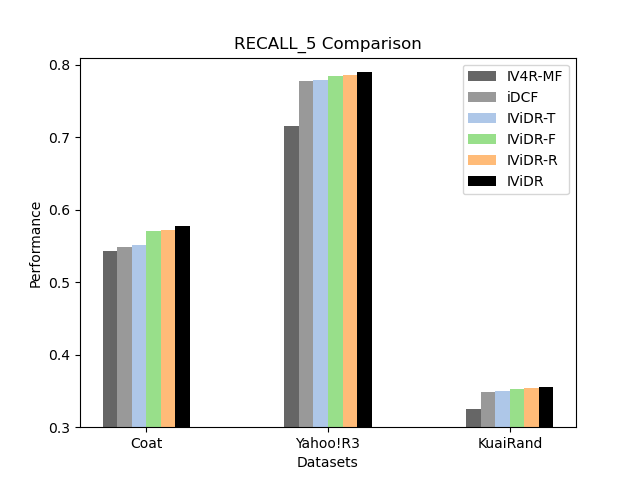}
         \captionsetup{labelformat=empty} 
    \end{minipage}
    
    \caption{IViDR-T, IViDR-F, IViDR-R, and IViDR algorithms' recommendation performance on the Coat, Yahoo!R3, and KuaiRand datasets.}
    \label{fig:6}   
    \vspace{-10pt}
\end{figure*}

\begin{figure}[H]
    \centering
    \begin{minipage}[a]{0.46\linewidth}
        \centering
        \includegraphics[width=\linewidth]{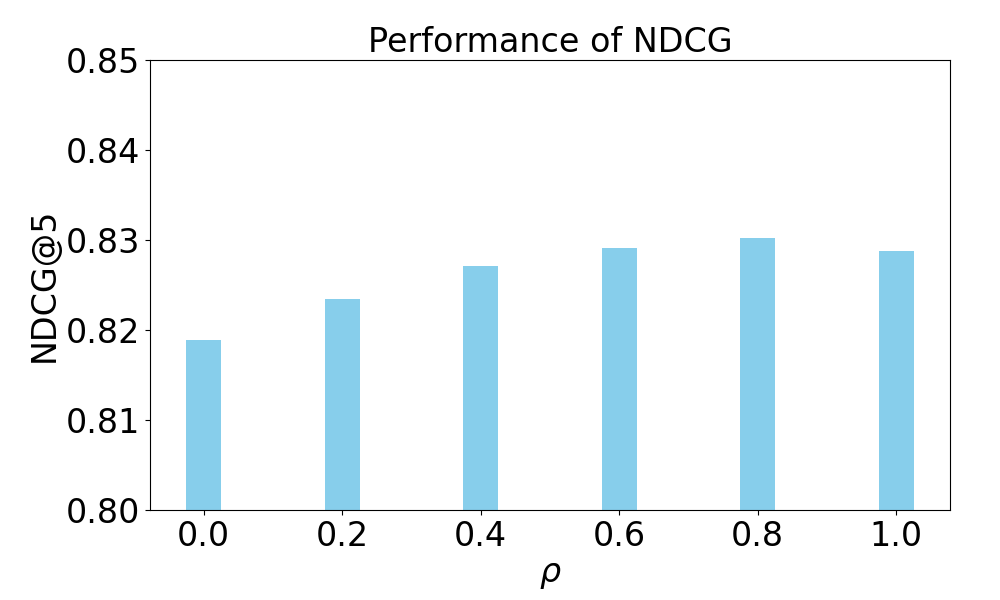}
         \captionsetup{labelformat=empty} 
        \caption*{(a) Effect of the $\rho$ selection. We
        show the results of NDCG@5 on the simulated datasets.}
        \vspace{0cm}  
    \end{minipage}
    \hfill
    \begin{minipage}[a]{0.46\linewidth}
        \centering
        \includegraphics[width=\linewidth]{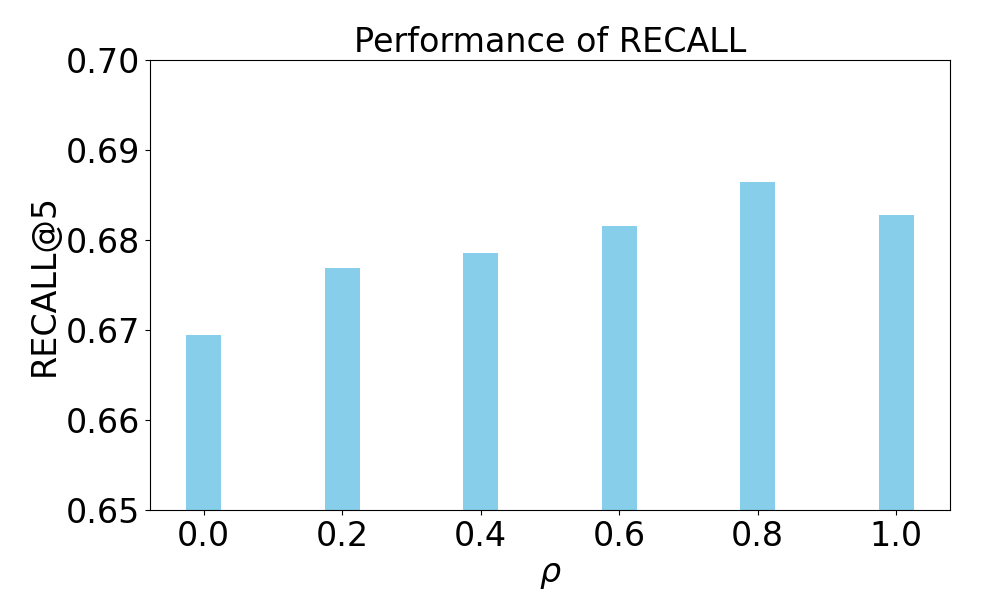}
         \captionsetup{labelformat=empty} 
        \caption*{(b) Effect of the $\rho$ selection. We
        show the results of RECALL@5 on the simulated datasets.}
    \end{minipage}
    \hfill
    \begin{minipage}[b]{0.46\linewidth}
        \centering
        \includegraphics[width=\linewidth]{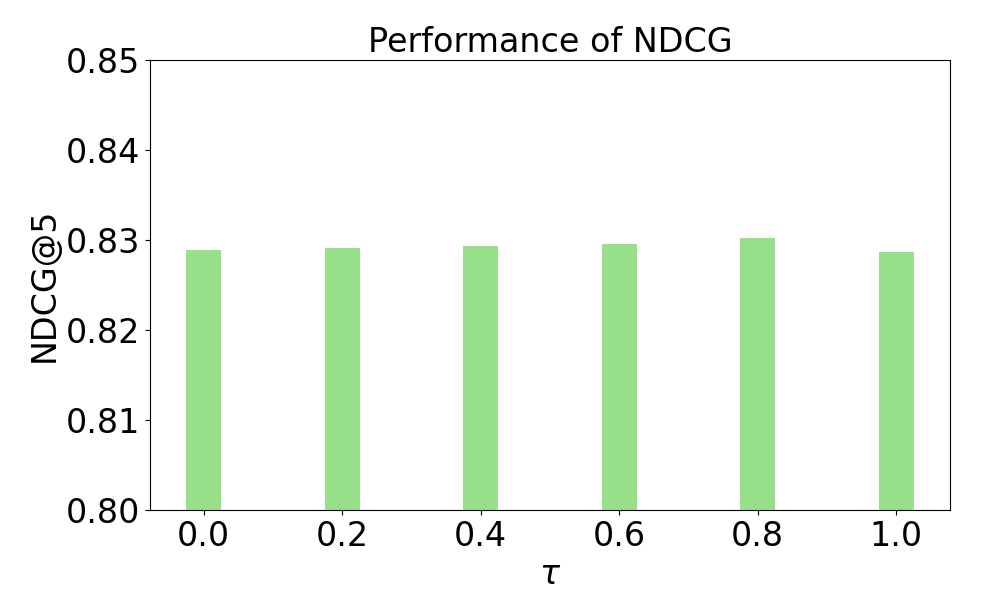}
         \captionsetup{labelformat=empty} 
        \caption*{(c) Effect of the $\tau$ selection. We
        show the results of NDCG@5 on the simulated datasets.}
        \vspace{0cm}
    \end{minipage}
    \hfill
    \begin{minipage}[b]{0.46\linewidth}
        \centering
        \includegraphics[width=\linewidth]{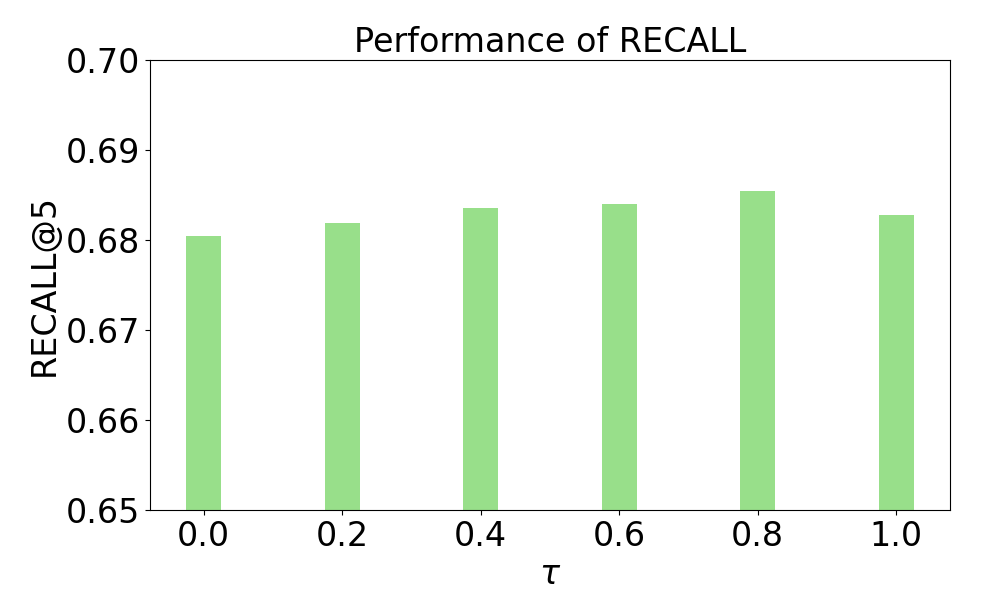}
         \captionsetup{labelformat=empty} 
        \caption*{(d) Effect of the $\tau$ selection. We
        show the results of RECALL@5 on the simulated datasets.}
    \end{minipage}
    \hfill
    \caption{ Effect of the hyper-parameter selection.}
    \label{fighp}
\end{figure}

\subsection{Ablation Study}
\label{appendix:E1}

We conducted an ablation study to show the effectiveness of the proposed architecture.
(1) IV4R-MF: IV4Rec applies IVs to decompose input vectors within recommendation models to mitigate the effects of latent confounders. IV4R-MF uses MF as the backbone.
(2) iDCF: The iDCF incorporates proximal causal inference techniques to identifiably predict users' counterfactual feedback on items. It predicts more accurately than Deconfounder.
(3) IViDR-T: The case where the original treatment $\mathbf{T}$ is added to the interactional data $\mathbf{X}$. 
(4) IViDR-F: The case where only fitted part $\widehat{\mathbf{T}}$ is added to the interactional data $\mathbf{X}$.  
(5) IViDR-R: The case where only residual part $\widetilde{\mathbf{T}}$ is added to the interactional data $\mathbf{X}$.

The experimental results are shown in Table \ref{tab:as}:  
With the addition of the original treatment $\mathbf{T}$, IViDR-T slightly improved the iDCF results on all three real datasets, but the difference in results was small.

With the addition of the fitted part $\widehat{\mathbf{T}}$, IViDR-F outperforms the iDCF results on all three real datasets.
Specifically, for the Coat dataset, NDCG@5 improved from 0.5638 to 0.5846, and RECALL@5 improved from 0.5493 to 0.5710. 
For the Yahoo!R3 dataset, NDCG@5 improved from 0.6410 to 0.6552, and RECALL@5 improved from 0.7780 to 0.7849. For the KuaiRand dataset, NDCG@5 improved from 0.4080 to 0.4130, and RECALL@5 improved from 0.3481 to 0.3529.

With the addition of the residual part $\widetilde{\mathbf{T}}$, IViDR outperforms IViDR-F on all three real datasets.
Specifically, for the Coat dataset, NDCG@5 improved from 0.5846 to 0.5903, and RECALL@5 improved from 0.5710 to 0.5783. 
For the Yahoo!R3 dataset, NDCG@5 improved from 0.6552 to 0.6602, and RECALL@5 improved from 0.7849 to 0.7901. For the KuaiRand dataset, NDCG@5 improved from 0.4130 to 0.4161, and RECALL@5 improved from 0.3529 to 0.3549.

The recommendation performance of the IViDR-T, IViDR-F, IViDR-R, and IViDR algorithms on the Coat, Yahoo!R3, and KuaiRand datasets is shown in Figure \ref{fig:6}.

\subsection{Hyper-parameter Analysis}
\label{appendix:E2}

We perform experiments to analyze the effect of the key hyperparameters of our method on the recommended performance.

Particularly, we fix all the other parameters (fixed $\alpha$=0.1, $\beta$=2.0, $\gamma$=0.0), then set the $\rho$ and $\tau$ in Eq. \cref{eq:15} from 0 to 1 with an interval of 0.2. 
As shown in Figure \ref{fighp}, we can see that:
(1) As the $\rho$ increases, NDCG@5 and RECALL@5 first get better and then decrease.
This means that appropriate weighting of the debiased interactional data $\mathbf{X}^{\mathrm{re}}$ can improve model recommendation performance.
(2) As the $\tau$ increases, NDCG@5 and RECALL@5 do not change much, but the general trend is still upward and then downward.
This means that appropriate weighting of the interactional data $\mathbf{X}$ can also improve model recommendation performance.

\end{document}